\documentclass[aps,prx,reprint,floatfix,superscriptaddress]{revtex4-2}
\usepackage{amsmath,amssymb,amsthm,bm,graphicx}
\usepackage{microtype}
\usepackage{silence}
\usepackage{hyperref}
\hypersetup{colorlinks=false, pdfborder={0 0 0}}
\graphicspath{{figures/}}

\newtheorem{theorem}{Theorem}
\newtheorem{proposition}[theorem]{Proposition}
\newtheorem{lemma}[theorem]{Lemma}
\newtheorem{corollary}[theorem]{Corollary}

\newtheorem{remark}[theorem]{Remark}
\newtheorem{open}{Open Problem}

\begin{document}
\title{Correlations decide a shallow-circuit advantage}
\author{Zijian Gong}
\affiliation{Southern University of Science and Technology, Shenzhen, Guangdong 518055, China}
\affiliation{Faculty of Computility Microelectronics, Shenzhen University of Advanced Technology, Shenzhen, Guangdong 518107, China}
\affiliation{Guangdong Provincial Key Laboratory of Computility Microelectronics, Shenzhen, Guangdong 518107, China}
\author{Zhaobin Lyu}
\affiliation{Pengxin Quantum Technology (Shenzhen) Co., Ltd., Shenzhen 518107, China}
\author{Jingjing Hu}
\email{hujingjing@suat-sz.edu.cn}
\affiliation{Faculty of Computility Microelectronics, Shenzhen University of Advanced Technology, Shenzhen, Guangdong 518107, China}
\affiliation{Guangdong Provincial Key Laboratory of Computility Microelectronics, Shenzhen, Guangdong 518107, China}
\affiliation{Pengxin Quantum Technology (Shenzhen) Co., Ltd., Shenzhen 518107, China}
\author{Dengfeng Li}
\email{lidengfeng@suat-sz.edu.cn}
\affiliation{Faculty of Computility Microelectronics, Shenzhen University of Advanced Technology, Shenzhen, Guangdong 518107, China}
\affiliation{Guangdong Provincial Key Laboratory of Computility Microelectronics, Shenzhen, Guangdong 518107, China}
\affiliation{Pengxin Quantum Technology (Shenzhen) Co., Ltd., Shenzhen 518107, China}
\author{Shuoming An}
\email{anshuoming@suat-sz.edu.cn}
\affiliation{Faculty of Computility Microelectronics, Shenzhen University of Advanced Technology, Shenzhen, Guangdong 518107, China}
\affiliation{Guangdong Provincial Key Laboratory of Computility Microelectronics, Shenzhen, Guangdong 518107, China}
\affiliation{Pengxin Quantum Technology (Shenzhen) Co., Ltd., Shenzhen 518107, China}
\date{\today}

\begin{abstract}
Quantum computers are claimed to produce samples ordinary classical computers cannot reproduce. The sharpest such claim known pits constant-depth quantum circuits, given one entangled state per run, against shallow classical circuits of few-input gates and limited randomness. No efficient test was known that certifies such a claim from classical samples alone. Here we show that what such a test must read is decided by the classical class's pattern of single-bit and pair correlations, not by a distance. The test rests on one inequality: any machine producing the samples sits no further from the target than the fraction it mislabels plus the deviation of its string half, the bits beyond the label, from uniform. A collapse theorem, machine-checked in the Lean~4 proof assistant, confines that second term, one setting of the sampler's random input bits at a time. Four natural checks provably fail; a fifth, reading those same correlations, catches a far-from-target construction they cannot reach, and is exact on the class of few-input gates. The label test is proven and sample-optimal in its own tolerance up to a logarithmic factor; the fifth check is sound, not only effective, on the bounded pinned-residue class, pairwise-uniform samplers whose residue, the string's weight modulo the prime, is pinned by few seeds and whose seed sharing is bounded. What binds an experiment is the $43$-qubit entangled resource state at fidelity near $0.99$, not the readout. What remains are two named open problems; a positive answer to the first would extend the guarantee to the full class.
\end{abstract}

\maketitle

\section{Introduction}\label{sec:intro}

A sampling-advantage claim is experimental, and verification, not the claim, is the credibility infrastructure of experimental science \cite{Huang2026vastworld}: the same asymmetry governs the energy scale that announces a particle and the symmetry that certifies a phase long after its order parameter is measured. Sampling advantage is where that infrastructure is thinnest: the check there is a list of output strings and nothing else, the device cannot be opened and its diagram cannot be trusted. The reviews find the existing checks resting on a proxy whose assumptions break or on a device-dependent certificate \cite{LaRose2026history,HangleiterEisert2023}. 

Cross-entropy benchmarking scores a device by how closely its samples match the ideal circuit. The debate over it, from the $2019$ Google claim \cite{Arute2019} and its hardness analysis \cite{BFNV2019} to the later peaked-circuit claims, their classical-simulation rebuttal \cite{KremerDupuis2026}, and the benchmarking proposals that followed \cite{Bentsen2026bench}, turned on exactly that. A classical algorithm reproduced $5$--$12\%$ of the cross-entropy scores those experiments reported \cite{GaoKalinowskiChouLukinBarakChoi2024}. The hardness of that score rests on XQUATH, the linear cross-entropy quantum threshold assumption, which posits that estimating a random circuit's output probabilities is hard on average and remains unproved \cite{AaronsonGunn2020}. A polynomial-time classical algorithm matches a noisy circuit as soon as the output spreads over many outcomes \cite{AharonovGaoLandauLiuVazirani2023}, a criterion the simulability literature now contests with one of its own, read off correlations instead \cite{Lee2025cmi}. A sampler and a verifier, both unconditional, reading classical samples alone would close the loop.

The task has the shape of a Bell test: the adversary is a model class, the prize an inequality no member of it satisfies. The class is the constant-locality classical samplers, each output bit a function of a fixed number of random input bits, and the witness is classical samples alone. Bene Watts and Parham supply the target \cite{BenewattsParham2026}: $D_n=(X,f(X))$ with $f=\mathrm{majmod}_p(X)\oplus\mathrm{parity}(X)$, the majority-modulo-$p$ of the Hamming weight XOR its parity, $p$ an odd prime, the tag's period in the weight being $2p$, sampled in constant depth within total-variation distance $1/6+O(n^{-c})$, where total variation is the largest probability gap two laws can put on one event and $n$ is the string length. Each run consumes one GHZ state, an entangled state whose qubits all flip together, and that state is the one experimental resource the verifier asks for. Bounded-fan-in classical circuits of bounded randomness cannot beat the trivial distance, the $1/2$ of a uniformly random string, at constant depth: their depth must grow at least as $\log\log n$. Ruling out a \emph{fixed} circuit takes a few samples; ruling out every NC$^0$-producible distribution -- everything the constant-depth circuits of bounded fan-in can output -- is, in their words, ``a harder task'' \cite{BenewattsParham2026}.

A sample is $(z,b)$, the string half $z$ with a tag $b$ computed from it, and the violation rate is $\beta(P)=\Pr[b\ne f(z)]$, the fraction of samples whose tag disagrees with the target's own. What classical samples alone can certify is the question.

Here we show that what decides such a test is the class's own correlation lattice, the single-bit and pair correlations sitting a power of two apart, rather than a distance. The test rests on one inequality in total variation:
\begin{equation}\label{eq:chain}
\mathrm{TV}\bigl(P,D_n\bigr)\;\le\;\beta(P)+\mathrm{TV}\bigl(P_z,U_{n-1}\bigr)
\end{equation}

A sampler far from the target is either wrong on many labels or far from uniform on its string half $P_z$, the marginal law of the string $z$ against the uniform string $U_{n-1}$. For an imposter, a sampler other than the target, at distance $\varepsilon$, the separation the verifier is built to detect, the inequality splits,
\begin{equation}\label{eq:dichotomy}
\begin{gathered}
\mathrm{TV}\bigl(P,D_n\bigr)\ge\varepsilon\;\Longrightarrow\;\beta(P)\ge\varepsilon/2\\
\qquad\text{or}\quad\mathrm{TV}\bigl(P_z,U_{n-1}\bigr)\ge\varepsilon/2 .
\end{gathered}
\end{equation}

The first case of Eq.~\eqref{eq:dichotomy} is proven unconditionally, at a system-size-independent budget and optimal up to a logarithmic factor, its lower bound the two-point Le Cam argument, which separates two nearby hypotheses; the second carries the difficulty, confined by a collapse theorem to low dimension, one fixing at a time, a fixing being one way of setting the seed bits. 

Four symmetric checks of that structure are provably insufficient, Sec.~\ref{sec:dichotomy}: an explicit $\beta=0$ construction, far from the target, clears at every polynomial budget the residue door, reading the law of the weight modulo $2p$, the modulus that carries the majority's input and the parity together; the rank door, reading the sample matrix's rank over $\mathbb F_2$, the two-element field whose addition is XOR; the weight-law door, reading the histogram of the weight; and anticoncentration, reading whether the string half spreads over many strings or concentrates on a few. The fifth reads a statistic that is quantized on the class, dyadic, a power of two apart, rather than a continuum. It is sound, not only effective, on the bounded pinned-residue class: samplers that are pairwise uniform, every single-bit and pair correlator vanishing, whose weight modulo the prime is already fixed by a bounded set of seed bits, and whose read is bounded either in fan-out, the number of outputs one seed feeds, or in the seeds that more than one output reads.

Unconditional separations can reach hardware demonstrations because they carry their own classical check \cite{Bravyi2018,BeneWattsKothariSchaefferTal2019,BGKT2020,CahaCoiteuxRoyKoenig2026}; a sampling claim carries none. The nearest separation is Grier \emph{et al.}'s \cite{GrierKaneMorrisOstuniWu2026}, and Kane, Ostuni and Wu put the hardness in the $z$-marginal itself \cite{KaneOstuniWu2024hamming}, where $D_n$'s marginal is uniform and the difficulty is the constraint $b=f(z)$. 

The class's further constructions and lower bounds \cite{LovettViola2012,CaoEisert2026,Zhang2024,Viola2012complexity} delimit what it produces, as do the structure theorems \cite{KaneOstuniWu2025symmetric} of Secs.~\ref{sec:dichotomy} and \ref{sec:perfixer}, and the line has been sharpened from the sampler's own side, dyadic product masses being samplable in constant depth where the non-dyadic ones provably are not \cite{HoracsekLeeShinkarViolaZhou2026}. The same program separates against other classical models, biased threshold circuits \cite{deOliveira2025qudit} and bounded-resource language models \cite{Arunachalam2026llm} among them.

Soundness against every NC$^0$ sampler, classical samples alone, no device model, no unproved conjecture: four conditions define the task. Every known verifier drops at least one of them, requiring instead a sub-logarithmic-depth device \cite{CoudronStarkVidick2021}, an entropy assumption \cite{BFNV2019}, device-dependent certificates \cite{Ringbauer2025,Martiel2026}, one-way functions \cite{BenedettiBuhrmanWeggemans2026}, a bounded distinguisher class \cite{BunDesaiPinto2026}, learning with errors \cite{Mahadev2018} or lattice knowledge \cite{Gheorghiu2026} -- the cryptographic structure such tests demand being itself a subject of study \cite{BartusekShalit2026}. The other verifiers give up the device model, the samples, or the proof of hardness: a shallow device \cite{Huang2024learning}, a peaked sampler verified by computing ideal outcome probabilities \cite{Deshpande2025peaked}, verification from the circuit rather than from its samples \cite{BravyiParhamTran2026}, a heuristic obfuscation whose hardness is conjectured \cite{Yan2026}. 

Only identity testing, checking samples against the known law directly, keeps all four, at $\Theta(2^{(n-1)/2}/\varepsilon^2)$ samples, exponential in the string length \cite{Canonne2022,Paninski2008}, with the junta bounds, on functions of few seeds, and the independence and probe-depth bounds, the latter restricting how many coordinates a query reads, placing that cost \cite{Blais2009,AliakbarpourBlaisRubinfeld2016,AlonAndoniKaufmanMatulefRubinfeldXie2007,Viola2023sampling} and $D_n$ itself classically samplable in depth $O(\log^2 n)$ \cite{Viola2012complexity,Czumaj2015}. Polynomial samples suffice under an efficiently samplable promise when inefficient verification is allowed \cite{CavalarGoldinGrayHirokaHsiehMorimae2025}, the no-structure route is unconditional relative to a random oracle \cite{YamakawaZhandry2022}, and the relation task has run at $105$ qubits by Google Quantum AI \cite{Kumar2025}.

The functional test is proven unconditionally, its proof machine-checked in the Lean~4 proof assistant, and a per-fixing collapse theorem reduces every sampler $S$ the test does not already reject, on a $1-\tfrac83\beta(S)$ fraction of fixings, to at most $d+O(p^2\log p)$ residue-relevant live bits, or to a residue parity that is a $d$-junta, a function of at most $d$ of the live bits, the bits a fixing leaves random; here the function is a parity, and $d$, the locality, is the number of seed bits each output reads. The fifth check proves soundness on a bounded pinned-residue class as well as rejecting. A $d$-local sampler routes its randomness through hub seeds of decreasing fan-out, seeds that many outputs read, and a read bounded in fan-out or in the seeds that more than one output reads cannot supply them. The residual is neither affine, the string half a linear function of its random bits plus a constant, nor symmetric, unchanged by permuting the string's coordinates: the affine case is proven at every constant locality, the symmetric one confined by Kane--Ostuni--Wu.

Two facts underpin the reduction. The first is that $f(z)$ depends only on $|z|\bmod 2p$, each $z$-bit having influence $1-2/p+o(1)$, that is, flipping it flips the tag with that probability, as Sec.~\ref{sec:struct} records. The second is that every sampler of App.~C of the Supplemental Material either has $\beta>3/16$, half the separation $\varepsilon=3/8$ set in Sec.~\ref{sec:prelim}, or fails a structure door, or is seed-limited. A structure door is one criterion a sampler must pass, and a seed-limited sampler draws on a bounded number of seeds that the anti-concentration door reaches only at a budget of order its support.

Sec.~\ref{sec:prelim} sets the model, Secs.~\ref{sec:struct}--\ref{sec:block} simplify the target, Sec.~\ref{sec:dichotomy} states the verifier and the open core, Sec.~\ref{sec:perfixer} proves the collapse, Sec.~\ref{sec:lean} reports the formalization, and Sec.~\ref{sec:open} closes.

\section{Preliminaries}\label{sec:prelim}
The verifier is defined against one fixed target family. Let $D_n=(X,f(X))$ with $X\sim\{0,1\}^{n-1}$ uniform and $f(x)=\mathrm{majmod}_p(x)\oplus\mathrm{parity}(x)$, where $\mathrm{majmod}_p(x)=\mathbf 1[|x|\bmod p>p/2]$ and $p$ is an odd prime; Watts--Parham take $p=n^c$, $c\in(0,1/2)$, in their sampling theorems and $p=\Theta(n^\alpha)$, $\alpha\in(\delta'/3,1/3)$, in the hardness regime \cite{BenewattsParham2026}. A sample is the $n$-bit string $(z,b)$ with $z=X$ and $b=f(X)$; each sample is a random string together with a tag, its label, computed from its weight, the number of ones. Shallow quantum circuits produce tag distributions that equally shallow classical circuits cannot imitate, and what this paper verifies is unconditional: the distance the target is sampled to is a constant, $1/6+O(n^{-c})$ for the GHZ-advice circuit, its advice the per-run GHZ state, and the thresholds need only any constant below $1/2$. In the numerical sections $n$ denotes the number of $z$ bits; the analytical statements write the same count as $n-1$.

A distribution $P$ on $\{0,1\}^n$ is \emph{NC$^0$-samplable} if $P=S(U_\ell)$ for $S:\{0,1\}^\ell\to\{0,1\}^n$ each of whose output bits $(z,b)$ depends on $O(1)$ seed bits. A sampler is \emph{$d$-local} when every output bit, the tag included, reads at most $d$ of them, so $d$ bounds each output separately and a claim of locality $d$ cannot be met by a tag reading more. Distances are total variation, $\mathrm{TV}(P,Q)=\max_A|P(A)-Q(A)|=\tfrac12\sum_x|P(x)-Q(x)|$. A \emph{verifier} takes $M$ independent samples and accepts or rejects; it is complete if it accepts $D_n$ with probability at least $1-\eta$, $\eta$ its allowed error, and sound if it rejects every NC$^0$-samplable $P$ with $\mathrm{TV}(P,D_n)\ge\varepsilon$ with probability at least $1-\eta$. We call an NC$^0$-samplable distribution an \emph{imposter}, the distribution-testing analogue of a spoofing distribution, and a test \emph{imposter-agnostic} if its definition does not depend on the structure of the imposter it rules out, so soundness, wherever established, is uniform over the class rather than certified construction-by-construction; the term is unrelated to agnostic testing.

Define the \emph{functional-violation rate}
\begin{equation}\label{eq:beta}
\beta(P):=\Pr_{(z,b)\sim P}[b\ne f(z)].
\end{equation}
Since $f$ is poly-time, $\beta$ is estimable to additive $\delta$ from $M=O(\delta^{-2}\log(1/\eta))$ samples, independent of $n$, and $\beta(P)\le\mathrm{TV}(P,D_n)$. The tolerance $\delta$ is the test's free estimation slack, fixed at the operating point to a third of the completeness margin, $\delta=(\varepsilon/2-\gamma)/3=1/144$ with $\gamma$ the device's own distance from the target: one third absorbs the estimator's own deviation, one the readout bias, and the third separates Test~1's rejection level from the dichotomy's bound, so that the two are exactly one tolerance apart. This $\delta$ is distinct from Watts--Parham's randomness exponent $\delta'$ and the per-group increments $\delta_j$ of Sec.~\ref{sec:perfixer}.

A \emph{junta} is a function of few seed bits, a $d$-junta of at most $d$. A \emph{fixing} assigns values to all but a \emph{live set} of seeds, those left random. For a sampler $S$ we write $\beta(S):=\beta(S(U_\ell))$, the violation rate of its output distribution.

The definition is deliberately information-theoretic, and one budget calibrates the window in which a physical device can use it. Completeness must tolerate the device's own distance $\gamma$ from $D_n$, so acceptance must extend to every distribution that close, and the window $(\gamma,\varepsilon/2-\delta)$ in which the device's imperfections and the adversary's freedom cannot be separated plays the role of the detection-loophole margin of a Bell test: the design is exactly tight at the stated $\gamma=1/6$, so the device's own $O(n^{-c})$ correction is charged to the readout.

At the operating point $\varepsilon=3/8$, the smallest dyadic fraction with denominator at most $8$ that clears $2\gamma=1/3$, Test~1, defined below, rejects at $\hat\beta>\varepsilon/2-\delta=13/72$ with $\hat\beta$ estimating $\beta$ of Eq.~\eqref{eq:beta}, a margin of $\varepsilon/2-\delta-\gamma=1/72$ over the stated device distance. Either way an honest device passes Test~1, while every imposter with $\beta\ge3/16$ dies at Test~1 alone. The account is anchored to that stated constant rather than to the weaker analyzed bounds, and a direct evaluation of the circuit corroborates that it is attained at accessible scale, giving $0.1677$ at the operating depth and approaching $1/6$ as the system grows; the arithmetic of the account, the readout budget, and that evaluation are App.~D.

Two features of the operating point are worth carrying here. The first is that noise enters at two places and \emph{the binding one is not the readout}: the certificate's requirement is the preparation and gate error of the $43$-qubit resource state, one qubit per output bit, whose fidelity must reach about $0.99$, against a readout requirement a factor of about five below present transmon performance, at a polynomial sample budget; prepared states at this scale have reached fidelities of $0.595(8)$ at $60$ qubits and $0.723(10)$ at $36$ \cite{Bao2024ghz}. The second is that the noise-optimal prime is $3$, where that budget is loosest, while the proven classical hardness of $D_n$ is stated for $p=\Theta(n^\alpha)$: the operating point calibrates the verifier's threshold rather than certifying an advantage at $p=3$ itself, and Test~1 and Test~2 are defined for every odd $p\ge3$. The class the verifier is defined against is NC$^0$ itself, Sec.~\ref{sec:dichotomy}.

\section{Structural simplification}\label{sec:struct}
The target function is globally sensitive yet weight-local: it reads $z$ only through $|z|\bmod 2p$, so a verifier can reduce the functional constraint to a single $\mathbb Z/2p\mathbb Z$ statistic.

\begin{lemma}\label{lem:struct}
$f(z)=F(|z|\bmod 2p)$ for $F:\mathbb Z/2p\mathbb Z\to\{0,1\}$, verified in Lean. Each bit of $z$ has influence
$\mathrm{Inf}_i(f)=\Pr_z[f(z)\ne f(z\oplus e_i)]=1-2/p+o(1)$.
\end{lemma}
\begin{proof}
$f(z)=\mathbf 1[|z|\bmod p>p/2]\oplus(|z|\bmod 2)$, and the pair $(|z|\bmod p,\,|z|\bmod 2)$ is determined by $|z|\bmod 2p$ for odd $p$. For the influence: flipping $z_i$ shifts $|z|$ by $\pm1$, hence $|z|\bmod 2p$ by $\pm1$, and $F$ flips at the flipped residue exactly when $\mathrm{majmod}$ does not, since the parity factor flips at every step. The step function $\mathrm{majmod}$ on $\mathbb Z/p$ jumps at exactly two of the $p$ adjacent pairs, hence at four of the $2p$ adjacent pairs of $\mathbb Z/2p$, and the parity factor at every one.  So $F$ jumps exactly where $\mathrm{majmod}$ is flat, at $2p-4$ of the $2p$ pairs, giving the exact residue-uniform influence $c_F=1-2/p$.  The jump count and the constant are verified in Lean. The binomial weight $|z|$ for uniform $z$ has residue within total variation $o(1)$ of uniform on $\mathbb Z/2p$, whence $\mathrm{Inf}_i(f)=c_F+o(1)$; this step is machine-checked end to end in the Fourier module, Sec.~\ref{sec:lean}; numerically $c_F=1/3,3/5,5/7$ for $p=3,5,7$, Fig.~\ref{fig:residue}.
\end{proof}
The influence in Lemma~\ref{lem:struct} is that statistic's jump count, and the chain is one line. Flipping a bit moves $|z|$ by $\pm1$ and hence the residue by $\pm1$, so the tag flips exactly at the adjacent residues where $F$ jumps. $F$ jumps exactly at the pairs where $\mathrm{majmod}$ is flat, the parity factor flipping at every step,
\begin{equation}
\begin{aligned}
\mathrm{Inf}_i(f)&=\Pr_z\bigl[F(|z|\bmod 2p)\ne F(|z|\bmod 2p\pm1)\bigr]\\
&=1-\frac{2}{p}\;=\;c_F ,
\end{aligned}
\end{equation}
the jump count being $2p-4$ of the $2p$ adjacent pairs of $\mathbb Z/2p$, $\mathrm{majmod}$ jumping at four of them and the parity factor at every one. And $f$ is \emph{globally sensitive}: no proper subset determines it. A concrete instance fixes the picture: at $p=3$, $n=5$, the tag is determined by the weight modulo $6$, a single number the verifier reads off every sample.

\begin{figure}[!htbp]
\centering\includegraphics[width=\linewidth]{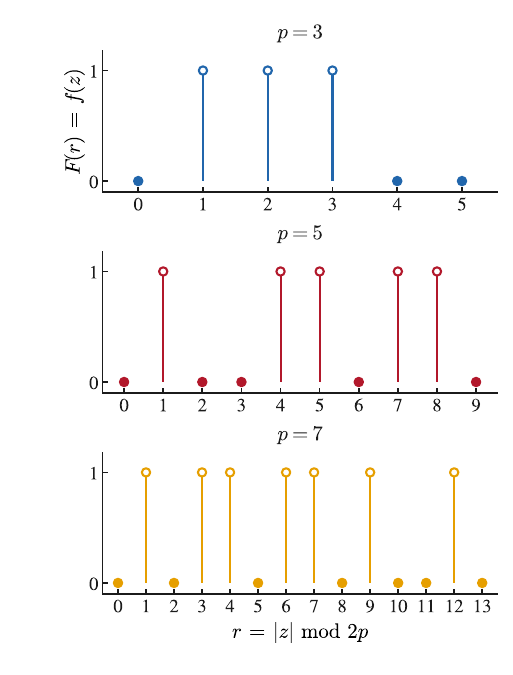}
\caption{Residue pattern of the target tag, the pattern a sample's tag is read from. One panel per prime, $p=3,5,7$. The horizontal axis is the residue, the Hamming weight $|z|$ of the sample's string half $z$ reduced modulo $2p$; the vertical axis is the tag $F(r)$, the answer the sample is meant to carry. Hollow markers are $F=1$ and filled markers $F=0$. The tag sees $z$ only through this pattern, which jumps at $2p-4$ of the $2p$ adjacent residues, and that jump count is the influence $1-2/p$ the rest of the paper reads: flipping one bit of $z$ changes the tag with exactly that probability.}
\label{fig:residue}
\end{figure}

\section{False shortcuts and the third door}\label{sec:block}
A natural hope is that small $\beta$ forces any NC$^0$ sampler to have small support, so that a collision test catches it; this fails. The $2p$-block sampler of App.~B achieves $\beta=0$ on a support of size $2^G$, exponential in the number $G$ of $2p$-blocks, each copying one seed bit, yet is caught by a single sample, and the escape sampler there carries $\beta=0$ at deficient rank; the two exhibit the failure of a support- or $\beta$-only argument and the need for the weight-histogram door (2c) and the rank door (2b) of Sec.~\ref{sec:dichotomy}.

\section{Dichotomy verifier and the open core}\label{sec:dichotomy}
The block-constancy structure that catches the $2p$-block sampler is a polynomial-time predicate of the sample; in general the verifier may compute \emph{any} polynomial-time statistic.

\begin{lemma}[TV chain rule]\label{lem:tvdec}
For any sampler $P=(z,b)$ and $D_n=(X,f(X))$,
\begin{equation}\label{eq:split}
\mathrm{TV}\bigl(P,D_n\bigr)\;\le\;\beta(P)+\mathrm{TV}(P_z,\,U_{n-1}),
\end{equation}
where $P_z$ is the $z$-marginal of $P$ and $U_{n-1}$ the uniform $(n-1)$-bit distribution. The chain rule is verified in Lean, as is the easy direction $\beta\le\mathrm{TV}$.
\end{lemma}
\begin{proof}
Insert the product distribution $P_z\otimes\delta_{f}$, with $\delta_{f}$ the point mass at $f(z)$, between $P$ and $D_n$ and apply the triangle inequality; the two resulting terms are $\beta$ and $\mathrm{TV}(P_z,U_{n-1})$. The full two-line argument is in App.~A.
\end{proof}

\begin{corollary}[Dichotomy]\label{cor:dich}
Every imposter with $\mathrm{TV}(P,D_n)\ge\varepsilon$ satisfies $\beta(P)\ge\varepsilon/2$, or $\mathrm{TV}(P_z,U_{n-1})\ge\varepsilon/2$: the two cases of Eq.~\eqref{eq:split}, the first decided by Test~1 alone and the second the residual that the rest of the paper locates.
\end{corollary}

\begin{corollary}[Soundness border]\label{cor:soundborder}
The verifier is sound against every NC$^0$-samplable imposter whose $z$-marginal is either within $\varepsilon/2$ of uniform or seed-limited. Such an imposter with $\mathrm{TV}(P,D_n)\ge\varepsilon$ has, by Cor.~\ref{cor:dich}, $\beta(P)\ge\varepsilon/2$. Test~1 rejects it at $O(\delta^{-2}\log(1/\eta))$ samples, its level one tolerance below the dichotomy's bound. Otherwise $\beta<\varepsilon/2$, the dichotomy gives $\mathrm{TV}(P_z,U_{n-1})>\varepsilon/2$, and a sampler of $z$-support at most the budget's own $M-T$, such as a $\le\log_2(M-T)$-seed NC$^0$ sampler, is caught deterministically by door (2d) at $S_z+T\le M$, $S_z$ the sampler's $z$-support and $T$ the door's threshold. Soundness therefore fails only on the residual case: a $z$-marginal far from uniform yet spread enough that the anti-concentration door's collision test stays near-uniform. That case is occupied by the paper's searches, and what occupies it is caught, though not by the four symmetric doors.
\end{corollary}

The seed-gap ladder of App.~C is graded by seed count; its $n_s=12$ rung, twelve seed bits and $g=16$ further $z$-bits computed from them, carries $\beta=639/4096=0.1560\le\varepsilon/2-2\delta=25/144$, so Test~1 accepts it. Its residue TV $0.0261$ and weight-histogram TV $0.0378$ lie below the two thresholds and its rank is $28$ of $28$, so the three structure doors leave it clean. Its $z$-support $4096$ inside the rung's $28$ $z$-bits puts $\mathrm{TV}(P_z,U_{28})\ge1-2^{-16}$, hence $\mathrm{TV}(P,D_n)\ge\varepsilon$ by projection.

Two facts close it, both in App.~C. Door (2d) rejects it deterministically on the shared stream that Test~1 needs anyway, and reaches the ladder's deeper rungs, whose support outgrows that stream, at a budget of order their own support; and independently of any budget the rung reads twelve of its twenty-eight output bits through bare \textsc{and} and \textsc{or} gates, so those bits carry single-bit biases, their probabilities of a one, of $1/8$, $1/4$, $3/4$ or $7/8$ rather than $1/2$, which door (2e), reading coordinates the four doors ignore, rejects from three hundred samples. That catch is not the rung's alone: over the ladder's fifty-one endpoints and their four hundred and twenty-one gate lists, every one carries a single-bit correlation of at least $1/2$, so door (2e) rejects each without any budget at all, App.~C. Directly measuring $\mathrm{TV}(P_z,U_{n-1})$ still costs $\Theta(2^{(n-1)/2}/\varepsilon^2)$, Open Problem~\ref{open:dist}; closing this class is the low-degree door rather than a budget.

\begin{figure}[!htbp]
\centering\includegraphics[width=\linewidth]{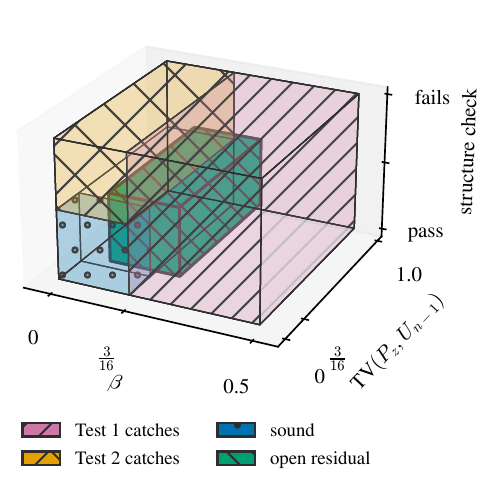}
\caption{The verifier as a map of where an imposter lands. Three coordinates place it: how often its answer is wrong, $\beta=\Pr[b\ne f(z)]$; how far its string half is from uniform, $\mathrm{TV}(P_z,U_{n-1})$; and whether a structure check fires, a schematic coordinate that orders the regions and carries no scale. Here $D_n=(X,f(X))$ is the target law, $X$ uniform on the string cube and $f$ the tag of Fig.~\ref{fig:residue}, $P$ the imposter's law, $P_z$ its string half and $U_{n-1}$ the uniform law. Test~1 rejects the wrong-answer region above its level $13/72$; Test~2 rejects the wrong-shape region; and a device whose string half is near uniform passes. The heavy-edged region, few wrong answers together with a string half far from uniform and the shape checks passing, is what the verifier cannot yet close, the open residual of Open Problem~\ref{open:dist}. Two caveats on the axes: the $\beta$ axis is truncated at $0.5$ and its $3/16$ tick is the dichotomy bound $\varepsilon/2$ rather than Test~1's level, and the $3/16$ tick on the distance axis is that same number read as a distance, the soundness border.}
\label{fig:verifier}
\end{figure}

Together they give a two-test verifier, shown as the decision map of Fig.~\ref{fig:verifier}:
\textbf{Test 1, the $\beta$-test, proven.} Estimate $\beta$ from $M=O(\delta^{-2}\log(1/\eta))$ samples at estimation tolerance $\delta$; reject if $\hat\beta>\varepsilon/2-\delta$, so that a sampler with $\beta\ge\varepsilon/2$, the dichotomy's own bound, is caught with probability at least $1-\eta$. This is imposter-agnostic: it recomputes the fixed function $f$ and posits nothing about the imposter; it catches the first case of the dichotomy, the Test-1 rejection region of Fig.~\ref{fig:verifier}.

\textbf{Test 2, the $z$-marginal structure test, candidate.} On the residual case $\beta<\varepsilon/2$, run poly-sample tests, each agnostic in its definition, a \emph{door} being one rejection criterion a sampler must pass to be accepted. The structure test is the union of five such doors, (2a) residue-uniformity, (2b) $\mathbb F_2$-rank, (2c) weight-histogram, (2d) anti-concentration, and (2e) low-degree correlators; Sec.~\ref{sec:perfixer} explains which constructions force each, and the open question is whether the union covers every residual imposter. Four of the five are \emph{symmetric} statistics of the $z$-marginal, invariant under permuting the $z$-coordinates, and those four are provably not sufficient on their own, Sec.~\ref{sec:dichotomy} below and App.~C; the fifth is not a symmetric statistic of the marginal but a coordinate-indexed one, and what it reads lies on the class's dyadic lattice.

\emph{(2a) Residue uniformity:} estimate the law of $|z|\bmod 2p$ and reject if it is $\tau$-far from uniform in total variation, at the optimal uniformity-testing rate for that alphabet \cite{Paninski2008,Canonne2022}, quoted in App.~D.

The threshold must clear the device's own deviation: by data processing the device's residue marginal deviates by at most $\gamma+o(1)$ in the worst case of the device-distance class, and the window $(\gamma,\varepsilon/2)$, nonempty since $\varepsilon>2\gamma$, separates that worst case from the idealized imposters of Prop.~\ref{prop:mixture}, which force $\varepsilon/2-o(1)$. We take $\tau=7/40$ at the operating point, and exact uniformity of the device's $z$-marginal, which would make any $\tau>0$ feasible at the specific-device promise, is not achieved: the unitary realization perturbs the ideal marginal by $O(n^{-c})$ in total variation, so the data-processing bound is the correct statement, while the exact target's own residue marginal deviates by only the character decay of Prop.~\ref{prop:decorr}, $0.050$ at the numerical depth $n-1=17$ and $1.6\times10^{-3}$ at the operating one, both far below $\tau$.

\emph{(2b) $\mathbb F_2$-rank:} from $M$ samples form the $(n-1)\times M$ binary matrix of $z$-bits and reject if its rank over $\mathbb F_2$, the two-element field whose addition is XOR, is below $n-1$, the rank computed by mod-2 Gaussian elimination. Under $D_n$ the deficiency probability is $\le\eta$, so $M$ must exceed $n-1$ by a logarithmic margin and no more, App.~C. This completeness statement is for $D_n$ itself, the rank door making no promise about the worst case of the device-distance class, while the analyzed circuit's $O(n^{-c})$-near-uniform marginal passes it with probability $1-\eta-O(n^{-c})$.

\emph{(2c) Weight-histogram:} compare the Hamming-weight law of $z$ with $\mathrm{Binomial}(n-1,\tfrac12)$, the exact law under $D_n$, and reject if it is $\tau'$-far in total variation. The residue is a function of the weight, so up to a binomial correction this door separates every distribution the residue door separates, and weight-distorted distributions besides, those whose weight law departs from the binomial; the correction is the distance from uniform of $\mathrm{Binomial}(n-1,\tfrac12)$ modulo $2p$, $1.6\times10^{-3}$ at $n-1=42$, against a door separation of $\tau-\tau'=0.115$, so the domination holds where the verifier is used. Its threshold stands below the smallest weight-law distance over the surviving family, and it is not an anti-concentration barrier against over-dispersion, App.~C.

The composite completeness statement is thereby: Test~1 and Test~2a accept every distribution within $\gamma$ of $D_n$, the analyzed circuit among them, while Test~2b's completeness is anchored to the circuit's marginal rather than the whole $\gamma$-class, so ``an honest device passes'' inherits the circuit promise rather than the class worst case. The doors are complementary: a tiny-support sampler with $z\in\{0,e_1,\dots,e_{n-1}\}$ has $\beta=0$ and full rank but concentrates the residue and dies on (2a), its weight law, concentrated at weights $\{0,1\}$, being likewise far from $\mathrm{Binomial}(n-1,\tfrac12)$ and caught on (2c); the $2p$-block and escape samplers of Sec.~\ref{sec:block}, by contrast, have uniform-enough or concentrated residues at deficient rank and die on (2b).

\emph{(2d) Anti-concentration, the $z$-support deficit:} from $M$ samples count $D$, the number of distinct $z$ values, and reject if $D\le M-T$. A law with $z$-support at most $S_z$ has $D\le S_z$, so once $S_z+T\le M$ the door rejects it deterministically, and the threshold is a free budget parameter: any $T$ above the near-uniform collision count $M^{2}/2^{n-1}$ preserves completeness, and the operating point takes $T=M/2$, which clears that count at the operating depth while at the small depths of App.~C the door is run with the threshold its own null requires. The same rejection reads as a threshold on the collision-pair count $\sum_{i<j}[z_i=z_j]$, which is at least $T$ whenever the door fires. The statistic is the classical collision count of uniformity testing, so the door's price is set by the alphabet rather than by the statistic, and its completeness holds at every depth where $M\ll2^{(n-1)/2}$, the operating point among them. Door (2d) is the cheap complement the structure doors lack: it reads the $z$-support directly, where doors (2a)--(2c) read only proxies, and the catch side is machine-checked.

\emph{(2e) Low-degree correlators:} from $M$ samples estimate the single-bit correlations $\mathbb E(-1)^{z_i}$ and the two-point correlations $\mathbb E(-1)^{z_i+z_j}$, $i\ne j$, of the $z$-coordinates, and reject if any estimated coefficient exceeds a threshold $\theta_{\mathrm{c}}$ placed above the uniform null's sampled upper quantile. The door is \emph{exact} on the constant-locality class, for a reason no symmetric statistic of the marginal has: the quantity it reads is quantized there. A $z$-coordinate of a $d$-local sampler is a Boolean function of at most $d$ uniformly distributed seed bits, so its single-bit correlation is one plus an integer multiple of $2^{-(d-1)}$, and the two-point correlation of a pair is likewise one plus an integer multiple of $2^{-(2d-1)}$. Every coefficient on the class is therefore \emph{dyadic}, and the lattice spacing is a property of the class rather than of the budget: at $d\le5$ the smallest nonzero single-bit correlation is $2^{-4}=1/16$, and the two-point lattice is finer, $2^{-(2d-1)}$, $2^{-5}$ at $d=3$ and $2^{-7}$ at $d=4$, so the spacing the door must resolve shrinks with the locality. The door decides by lattice membership wherever its threshold sits below that spacing and the budget resolves it, which is what no statistic of the marginal can do, and the counting identity behind the lattice is machine-checked. 

An imposter that is not of constant locality has no reason to be dyadic, and a non-dyadic product target is provably hard to sample even for bounded-degree polynomial sources with unbounded seed, where the overlap with such a target decays exponentially \cite{KhodabandehShinkar2026,Zhong2026}; one that is must exhibit one of the finitely many dyadic values, which the door distinguishes from zero at a sample cost polynomial in the system size, wherever the budget resolves the spacing. Under $D_n$ every one of these coefficients vanishes exactly, the $z$-coordinates being uniform and independent, so the door consumes no completeness budget beyond its own sampling error.

The door is the one the four symmetric statistics cannot replace: a sampler can be tuned to sit just under any threshold on an approximate statistic, while either the coefficient lies on the correlation lattice or the sampler is outside the class.

Door (2d)'s reach is a budget rather than a class bound. The passers of App.~C are seed-limited at their scans' seed budget, so it rejects each deterministically there, at a budget far below the $\Theta(2^{(n-1)/2}/\varepsilon^{2})\approx1.5\times10^{7}$ identity-testing cost at the operating depth $n-1=42$, Open Problem~\ref{open:dist}. What it does not catch is a many-seed imposter whose large-support $z$-marginal is non-uniform yet spread, with collision excess below the threshold, and that case is the residual below. This is not the failure Sec.~\ref{sec:block} records: there a \emph{support-} or $\beta$\emph{-only} argument breaks, whereas door (2d) is the support-deficit test of the \emph{three-door passers}, the $\varepsilon$-far samplers with good residue, full rank and binomial weight law but tiny support that (2a)--(2c) miss entirely.

\paragraph{The four symmetric doors are not sufficient.} At the thresholds stated above, doors (2a)--(2d) together are refuted by an explicit construction, App.~C. Take $601$ $z$-bits. Sixteen of its $272$ seeds enter the weight with increment $6\equiv0\pmod{2p}$ and so are invisible to the tag. The tag itself is a fixed Boolean function of five further seeds, four of them used as unit increments of the weight and one entering three times, with two control bits making every coordinate a four-way multiplexer, so each output reads at most six seeds. The resulting sampler has $\beta=0$ \emph{exactly}, by construction, and its $z$-marginal, the image of $272$ random bits, sits on at most $2^{272}$ of the $2^{601}$ strings, at total variation at least $1-2^{-329}$ from uniform: it is $\varepsilon$-far from $D_n$. Its residue TV is $0.0417$, below $\tau=7/40$; its weight-histogram TV is $0.0529$, below $\tau'=0.06$; its $z$-sample matrix has $\mathbb F_2$-rank $601$ of $601$; and its $z$-support is at least $2^{270}$, two of the $272$ seeds being redundant, so the anti-concentration door's collision statistic is at most $4\cdot2^{-270}$ and no polynomial budget reaches it. The construction therefore clears all four symmetric doors at every polynomial budget, and it is no corner artifact: it clears them at two tag read sizes, five and seven.

The construction falls at the low-degree door and nowhere else, its $1945$ nonzero two-point correlations reading $0.250$ to $0.500$, on the class's lattice at the coarse spacings its four-way multiplexers produce, against the uniform null's largest such coefficient of $0.079$ at $M=4000$ samples. Door (2e) at threshold $1/8$ rejects it in $20$ of $20$ runs at $M=2000$ and again at $M=4000$.

\paragraph{The thresholds are calibrated by the null.} Two of the five doors carry a total-variation threshold rather than a structural criterion, and both are set by measurement rather than by formula. The residue door's $\tau$ has a second role, the refutation of the device promise, and stays at $7/40$, above the class radius $\gamma$; the weight-histogram door's $\tau'$ carries no promise and is calibrated by the null's upper sampling quantile, which is computable because the null's weight law is $\mathrm{Binomial}(n-1,\tfrac12)$ and its residue law exactly uniform. The stated $\tau'=0.06$ therefore accepts weight-law distortions several times the sampling noise of its own null, while $\tau=7/40$ sits three times above the residue deviation at the tightest numerical depth and a hundred times above it at the operating one, both from the character decay of Prop.~\ref{prop:decorr}. The low-degree door's threshold is calibrated by the null as well, but the quantity it thresholds takes dyadic values on the class. The calibration at each depth, and its depth dependence, are App.~C.

\paragraph{What the fifth door proves, not only what it rejects: rigidity on the bounded pinned-residue class.}
What the door proves stands on its own. Call the string half \emph{pairwise uniform} if $\mathbb E(-1)^{z_a}=0$ for every output $a$ and $\mathbb E(-1)^{z_a+z_b}=0$ for $a\ne b$; under $D_n$ the $z$-coordinates are independent and balanced, so the target has the property, and a sampler that has it is one whose single-bit and two-point correlators all vanish.

\emph{Residue rigidity.} Let the string half be pairwise uniform on $\ell$ fair seed bits, and let $|z|\bmod p$ be a function of the seeds in a set $K$ of size $k$. For an output $a$ and an assignment $\sigma$ of those seeds put $A_a(\sigma)=\sum_{\omega:\,s_K(\omega)=\sigma}(-1)^{z_a(\omega)}$. Then for every $a$,
\begin{equation}\label{eq:rigidity}
\sum_\sigma r(\sigma)\,A_a(\sigma)\;\equiv\;-2^{\ell-1}\pmod p ,
\end{equation}
where $r(\sigma)$ is the residue at $\sigma$, so $A_a\not\equiv0$: every output bit is biased at some fixing of $s_K$. The proof is three lines. Write $Y=\sum_b(-1)^{z_b}=N-2|z|$. Pairwise uniformity kills the first moment of every output and the cross terms of its second, so $\sum_\omega Y\,(-1)^{z_a}=2^\ell$; writing $|z|\bmod p=r(s_K)$ splits the other side into a sum over the $2^k$ assignments, $p$ is odd, and $p\nmid2^{\ell-1}$, so the sum cannot vanish. With $K$ empty this already says that a pairwise-uniform string half never has $|z|\bmod p$ constant.

\emph{Hubs are forced.} Two outputs that share no seed outside $K$ are independent once $s_K$ is fixed, so the inner product of their vectors is the mean of $(-1)^{z_a+z_b}$ -- the independence gives the factorization and pairwise uniformity gives the vanishing, and both halves of the hypothesis are used -- whence $A_a$ and $A_b$ are orthogonal; the congruence makes each nonzero, and balance makes each orthogonal to the all-ones vector too. Nonzero mutually orthogonal vectors in $\mathbb R^{2^k}$ that all lie orthogonal to a nonzero vector are at most $2^k-1$ in number, so among the sampler's $N$ outputs, its output bits, at most $2^k-1$ have pairwise disjoint seed sets outside $K$. A $d$-local sampler therefore has some seed outside $K$ feeding at least $\bigl(N/(2^k-1)-1\bigr)/d$ outputs, and adjoining that seed to $K$ halves the bound, so the $j$-th hub feeds at least $\bigl(N/(2^{k+j}-1)-1\bigr)/d$ outputs and the hierarchy ends after about $\log_2(N/d)-k$ of them. This is the mechanism behind the signature the measurements report: the samplers that reach the residual route most of their outputs through the read's seeds, and every escape the searches built uses a fixing bit that gates many columns.

\emph{Many residues.} If instead the string half is uniform on every $2t$ of its bits, with $N\ge t$ and $t\le p-1$, then $|z|\bmod p$ takes at least $t+1$ values, and at $t=p-1$ all $p$ residues occur. The proof multiplies the vanishing polynomial $\prod_{\rho\in\Gamma}(|z|-\rho)$ of a putative $t$-element residue set $\Gamma$ by a character of degree $t$: the product has degree at most $2t$, so its mean is the coefficient of the empty monomial, and that coefficient is $\pm2^{-t}t!$ exactly, which an odd prime above $t$ cannot divide.

Together the three give a class on which Test~1 and door (2e) are \emph{sound}: against a $d$-local sampler bounded outside the tag's read in fan-out or in the seeds that more than one output reads, in the regime where $|z|\bmod p$ is pinned by a bounded seed set, the constant-locality read cannot meet the hub requirement. 

Three limitations are as tight as the lemmas' hypotheses. The lemmas are exact, so one seed on which the residue is wrong breaks the congruence, and they do not bound $\beta$ away from zero. They need $|z|\bmod p$ pinned rather than left free on the half-slice, the strings whose residue falls in one of the two level sets of $F$.  The slice leaves $|z|\bmod p$ free and ties it to the parity of $|z|$, a degree-$N$ statistic that a local sampler controls at no cost. And the hub conclusion is an existence statement about fan-out, contradicted only by such a bound: constant locality per output imposes neither, since one seed may feed every output while each still reads at most $d$ seeds. The lemmas therefore confine a class strictly smaller than the NC$^0$ class the verifier is defined against, and whether constant locality alone suffices is open. Carrying the congruence across fixings, where pairwise uniformity need not hold, is the robust problem.

\paragraph{Why the rank door is the right Test 2.}
The choice is a design heuristic. $\beta$ small forces the local output bit $g(r)$ to approximate the target tag at its own string half; since $f$ is globally sensitive, with total influence, the sum of the per-bit influences, equal to $\Theta(n)$ by Lemma~\ref{lem:struct}, and $g$ has bounded influence, the composition must collapse $\Theta(n)$ influence into $O(1)$. This holds only if the $z_i$ are strongly correlated so that their sum mod $2p$ is locally determined, pushing the $z$-image toward a low-$\mathbb F_2$-dimensional affine set, a rank-deficient sample matrix; making that implication precise is exactly Open Problem~\ref{open:dist}.

Viewed physically, the two doors read one constraint from two sides: the residue is the window through which local randomness leaks into the global statistic, the weight law, while the rank measures how much linear structure survives in the string half. The constructions of App.~C populate both ends of that trade-off, and none has reached the corner where low violation, full rank and an undisturbed weight law hold together.

Bounded-depth statistics are a priori insufficient: $\mathrm{AC}^0$, the constant-depth circuits of unbounded fan-in, is fooled by polylog-wise independence \cite{Braverman2010}, while the $\mathbb F_2$-rank of the sample matrix is not bounded-depth-computable, and identity against polynomial-size circuit statistics is testable at polynomially many samples \cite{BunDesaiPinto2026}. Finer statistics than the residue exist; the conjecture below is stated for the coarsest pair the constructions already defeat.

The constructed imposters confirm the pattern, their ranks tabulated in App.~C: each is rank-deficient, while in the scan of Fig.~\ref{fig:influence} no random local sampler reaches a small measured pool-violation, so the correlation the pattern names is what separates them.

The natural converse, that low $\beta$ forces a structure door, is false: a $\le3$-local sampler at $n-1=17$ has $\beta=0.1397$, full $\mathbb F_2$-rank, and a binomial weight law yet $z$-support $682$ of $2^{17}$, so $\mathrm{TV}(P_z,U_{n-1})=0.995$ at the soundness border of Cor.~\ref{cor:soundborder}. No door bridges the chain-rule gap of Eq.~\eqref{eq:chain}; the anti-concentration door catches the seed-limited passers on the shared stream of Sec.~\ref{sec:dichotomy}. That door is a budget rather than a class bound only above the supports of App.~C. The residual is the $z$-marginal's support-versus-uniform separation, and it has the small-support form that door (2d) recovers and the large-support form that Open Problem~\ref{open:dist} leaves open.

\paragraph{The verifier, assembled.} The decision rule is this: reject if $\hat\beta>\varepsilon/2-\delta$, or the residue marginal is $\tau$-far from uniform, or the sample rank is deficient, or the weight law is $\tau'$-far from the binomial, or the distinct-$z$ count is at most $M-T$, the anti-concentration door, or any single-bit or two-point correlation of the $z$-coordinates exceeds $\theta_{\mathrm{c}}$, the low-degree door. The tolerance $\delta$ is the estimation slack, so a sampler with $\beta\ge\varepsilon/2$ is rejected by Test~1 with probability at least $1-\eta$, and the rejection legs share the false-rejection budget.

The doors differ in what a rejection means, and the difference is the point. A Test-1 or residue rejection refutes the device promise $\mathrm{TV}(P,D_n)\le\gamma$ unconditionally, the residue threshold $\tau=7/40$ sitting above the class radius $\gamma$, whereas a rank, weight-histogram, or anti-concentration rejection is a structural diagnostic: the rank door makes no promise about the device-distance class, the histogram threshold $\tau'=0.06$ lies below $\gamma$ so a device inside the promised class may trip it, and door (2d) certifies only the seed-limited class, which it rejects deterministically. Door (2e) sits apart from both, a criterion exact on the constant-locality class rather than a diagnostic, since it reads a lattice that no symmetric statistic of the marginal reads.

The completeness side is quantitative at the operating point: the analyzed circuit passes with probability $1-\eta-O(n^{-c})\approx0.96$, Test~1's rejection level clearing the circuit's violation rate by about nine standard deviations at its budget, while the structure and anti-concentration doors reject it with probability at most $\eta+O(n^{-c})$ from its $O(n^{-c})$-near-uniform marginal. Test~1 and the residue door carry the whole device-distance class, the rank and histogram doors being anchored to the analyzed circuit. The operating point's symbol table, the per-door sample budgets and tolerant rates, and the standard-error arithmetic behind these numbers are App.~D.

\begin{figure}[!htbp]
\centering\includegraphics[width=\linewidth]{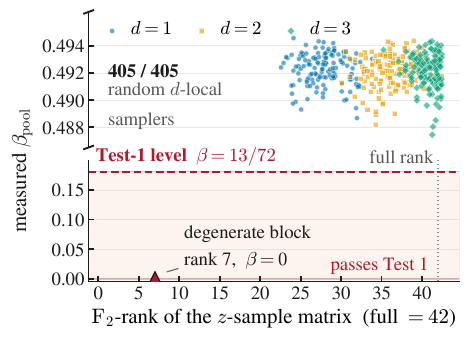}
\caption{On this class a low violation rate goes with rank deficiency. Each point is one of $405$ random $d$-local samplers, $d\in\{1,2,3\}$, at $42$ $z$-bits and $p=3$, each output bit an XOR of $d$ seed bits, each locality drawn with its own color and marker shape, plus the $2p$-block sampler, the construction of Sec.~\ref{sec:block}, as a triangle. The horizontal axis is the $\mathbb F_2$-rank of the sampler's $z$-sample matrix; the vertical axis is its violation rate against the best member of the $172$-predictor pool of App.~C, a cheap stand-in for the true $\beta$ whose members are the constants, the single seeds and their pairwise XORs over the first $18$ seeds, rather than against the target tag $f$. Random samplers sit at ranks $23$--$42$ with $\beta\ge0.487$; the only $\beta\approx0$ point is the degenerate block at rank $7$; and no full-rank sampler reaches $\beta<0.2$, the empty corner. The vertical axis is broken between $0.20$ and $0.4865$ and no point lies in the omitted band. The dashed line is Test~1's rejection level $13/72$.}
\label{fig:influence}
\end{figure}

The residual of the completeness question has two analytic components, both proven in Lean, Sec.~\ref{sec:lean}. Characters, the multiplicative Fourier modes $\omega^{kW}$ of the residue circle, are the right tool for the first because their expectations factor over independent blocks, which is what the fixing decomposition exploits. Throughout the sequel, a \emph{half-slice} is the set $H_c=\{z:F(|z|\bmod 2p)=c\}$ of strings whose residue falls in one of the two $F$-level sets, half of the $2p$ residues.

\begin{proposition}[Character-sum decorrelation, tensorized]\label{prop:decorr}
Let the block weights of a sampler's $z$-marginal be independent across blocks, as holds conditionally on a fixing of the live set of Sec.~\ref{sec:perfixer}, block $b$ contributing $W_b$ taking values in $A_b$. Write $\omega=e^{i\pi/p}$ for the primitive $2p$-th root of unity and $q_k=\mathbb E\,\omega^{kW}$ for the $k$-th Fourier mode of the residue marginal $W=|z|\bmod 2p$ on $\mathbb Z/2p$, the residue of a sample being its Hamming weight modulo $2p$; $\omega^{kW}$ is the $k$-th character of the cyclic group $\mathbb Z/2p$. Then the character factorizes: $q_k=\prod_b q_{k,b}$ with $q_{k,b}=\mathbb E\,[\omega^{kW_b}]$, and $|q_k|=\prod_b|q_{k,b}|$. Two decay bounds hold: \textup{(i)} if the weight law of $W_b$ lies within total variation $\delta_b$ of the adjacent pair $\{j_b,j_b{+}1\}$, then $|q_{k,b}|\le|\cos(\pi k/2p)|+2\delta_b$, since $|\omega^{kx}|=1$ caps the difference of the two expectations by twice their total variation; \textup{(ii)} if $W_b$ is uniform on the pair, $|q_{k,b}|=|\cos(\pi k/2p)|$ exactly. In particular $G$ independent balanced blocks give $|q_k|\le\cos(\pi/2p)^G=\exp(-\Theta(n/p^3))=o(1)$ at every $k\not\equiv0\pmod{2p}$, for $G=\Theta(n/p)$ and $p=o(n^{1/3})$, so the residue marginal is exponentially close to uniform on $\mathbb Z/2p$. A sampler supported on a half-slice instead concentrates its residue marginal on the $p$ residues of $F^{-1}(c)$, at total variation at least $1/2$ from uniform, and the two regimes are disjoint.
\end{proposition}

The unnormalized character sums behind (i) and (ii), and their tensorization, are the machine-checked statements; the normalized per-block modulus divides out the factor $|A_b|$.

\begin{lemma}[Modular half-slice subspaces are proper]\label{lem:halfspan}
Every $\mathbb F_2$-subspace $V\subseteq\mathbb F_2^{n-1}$ contained in the half-slice $H_0=\{z:F(|z|\bmod 2p)=0\}$ is proper, $V\ne\mathbb F_2^{n-1}$: the full space contains the unit vector $e_i$, of weight 1, and $F(1)=1\ne0$, verified in Lean. Proper subspaces do lie inside $H_0$, but never as a hyperplane: the two weight-4 vectors with supports $\{0,1,2,3\}$ and $\{0,1,4,5\}$ span one, machine-checked, while no hyperplane does, since at $p=3$ every hyperplane contains a vector of weight $1$ or $2$ and $F(1)=F(2)=1$. But any sampler whose $z$-support spans a subspace inside a half-slice is then rank-deficient and caught by the rank door.
\end{lemma}

With Prop.~\ref{prop:decorr} and Lemma~\ref{lem:halfspan} machine-checked, the formal residual of the completeness question is purely structural: either exhibit, for an arbitrary $\beta$-small NC$^0$ sampler, the independent block decomposition entering Prop.~\ref{prop:decorr}, for which Watts--Parham's own block machinery is the natural route \cite{BenewattsParham2026}, or show that its support lies in a proper $\mathbb F_2$-subspace, where Lemma~\ref{lem:halfspan} finishes. That analytic content carries no residual uncertainty.

\begin{open}[Residue, rank, weight-law, and $z$-marginal soundness]\label{open:dist}
Does a poly($n$)-sample test exist that accepts $D_n$ and rejects every NC$^0$-samplable $P$ with $\mathrm{TV}(P_z,U_{n-1})>\varepsilon/2$ and $\mathrm{TV}(P,D_n)\ge\varepsilon$? The question is the second case of the TV-chain-rule dichotomy, the $z$-marginal's support-versus-uniform separation, which the three structure doors measure only as proxies.

What occupies the question is known in some detail, and none of it is a counterexample. The passers a $\beta$-minimizing search finds are seed-limited: they hold all three structure doors while their $z$-marginals have support a small fraction of the cube, at total variation above $0.99$ from uniform, so door (2d) rejects each deterministically on the shared stream that Test~1 needs anyway. The realized set is not confined to shallow depths, an exact $1$-neighborhood ladder, a search that improves a sampler one gate swap at a time, reaching passers that clear the scan's four conditions at every depth probed, and seed lifts, re-optimizations at a larger seed count, holding every door above a support of $2^{12}$ in two independent basins. At the level of the $z$-marginal the same case is realized by an explicit joint law of $\beta=0$ that is $\varepsilon$-far while passing the residue, weight-histogram and anti-concentration doors at every depth. The corner it approaches, the half-cube marginal of maximal support $2^{n-2}$, is realizable by a one-local graph map, so building an imposter at that corner reduces to building a tag on it, and a hand proof makes the tag threshold uniform over corner maps and depths: a tag reaching $\beta\le3/16$ reads all but a bounded number of the seeds.

What remains open is a many-seed imposter whose $z$-marginal is far from uniform yet \emph{spread}, with a collision excess below the door's threshold, the collision excess being $2^{n-1}\sum_zP_z(z)^2$: one for the uniform marginal, larger the more the marginal concentrates, and at least $1+4t^{2}$ with $t=\mathrm{TV}(P_z,U_{n-1})$, so a marginal at distance $\varepsilon$ from uniform pays at least $25/16$, equality holding only for the two-level laws and attained there. Whether any affordable test catches such an imposter short of the $2^{(n-1)/2}/\varepsilon^2$ identity-testing block is the open residual, and the realized passers sit orders of magnitude above the floor their own $t$ allows. The wider premise $\beta(S)\le1/4$ has three-door passers with $\beta>3/16$ beyond the small coexistence depths. The census behind each of these statements is App.~C.

The residual is measured at a single $(\text{span},\text{rank},\text{alignment})$ point, at the four-door passers, App.~C. Its span, the width of the narrowest interval covering the conditional means $E[|z|\mid v]$ of the weight over the fixings, each fixing $v$ one cell of the conditioning, runs over $3.0$--$4.0$ there, a width its own dispersion bounds from below by the sharp span-versus-squared-deviation inequality \cite{KaiblingerSpangl2020}. That point carries $\mathbb F_2$-rank $=n-1$ \emph{exactly}, and its phase-aligned residual character sits below the one-sided sufficiency threshold of Open Problem~\ref{open:view} in five of the eleven read-based witnesses and above it in the other six, a split supported by positive cross-correlations.

What remains open is then two objects rather than four. The first is the $z$-marginal separation itself, Open Problem~\ref{open:dist}. The second is the \emph{view-average} object, Open Problem~\ref{open:view}: one class-separability requirement whose three faces are the across-fixing mean-span transfer, the cross-cell phase-diversity bound and the weight-side variance criterion, all reading the same cell-level quantity, a cell being one setting of the conditioned seeds.
\end{open}

\begin{open}[The view-average object: mixture transfer, phase diversity, and the weight-side variance criterion]\label{open:view}
Characterize exactly which bounded-fan-in samplers of locality at most $d$, the seed count unbounded, realize the residual of Open Problem~\ref{open:dist}: $\beta\le3/16$ with a $z$-marginal far from uniform, yet of collision-effective support $1/\sum_zP_z(z)^2$ beyond every anti-concentration budget, hence a fortiori of large $z$-support, near-Binomial and so evading the weight-histogram door, residue-near-uniform, and full $\mathbb F_2$-rank.

Three faces of that characterization are one obstruction rather than three: how much across-fixing mean spread a seed-limited $d$-local read can buy, how much of it survives the phase fan, the spread of the character's phase across cells, and how far below the binomial variance a histogram-clean read can sit. They are the \emph{mixture transfer}, that the three structure doors close the $\varepsilon$-close class at the operating point so that soundness reduces to Open Problem~\ref{open:dist}; the \emph{phase-diversity bound}, the cross-cell phase fan priced below; and the \emph{weight-side variance criterion}, the optimization-derived statement that no weight-histogram-clean sampler sits below the variance deficit $\delta^*(N)$, the shortfall of the weight variance below $N/4$. The law of total variance makes the across-fixing mean-spread term \emph{be} the mean span, the rank--span coupling is window-local rather than global, and the carrier of the phase-side residual is the same cell-level object the fan-priced sufficient condition reads, which is why the three stand or fall together. Each face, its criteria and its measurements are App.~C.

The residual is \emph{not} vacuous. The seed-gap ladder's $12$-seed rung of Sec.~\ref{sec:dichotomy} is an $\varepsilon$-far sampler that Test~1 accepts rather than merely fails to reject, and it leaves all three structure doors clean, so it shows the three $z$-structure doors to be insufficient on their own; it is not a witness of the residual, since door (2d) rejects it on the shared stream and the coordinate statistic rejects it independently. A sharper $\beta$ sits one rank dimension from the same corner, and the passers form a family rather than single points.

What is established is a pair of conditions on a cell-averaged residual-character magnitude. A \emph{necessary} one, machine-checked, is that $\beta\le3/16$ forces $\hat m_1:=\sum_vP(v)\|\chi_1(v)\|\ge\tfrac34(5/8-\tfrac13\hat m_3)$, where $\chi_k$ is the mod-$2p$ residual character of the view $v$, one cell of the conditioning, and $\hat m_3$ is the same average for $\chi_3$, so the residual needs a large cell-averaged residual-character magnitude. A \emph{sufficient} one needs no phase-alignment hypothesis, $\hat m_1\ge15/32+\hat D+\tfrac14\hat m_3\Longrightarrow\beta\le3/16$, with $\hat D$ the phase-fan deficit of the cells. The residual sits at the crossing of the two. The sharpest form is the two-channel criterion, which at $p=3$, where the spectrum of $(-1)^F$ has exactly the two odd channels, prices the second channel by its own sign: the per-cell conditional mean is then the sum of two \emph{real} channels, so $\tfrac43\hat M_1^{\mathrm a}+\tfrac13\hat m_3^{\sigma}\ge5/8\Longrightarrow\beta\le3/16$, where $\hat M_1^{\mathrm a}=\hat m_1-\hat D$ and $\sigma$ is the channels' relative sign, the one-sided condition being its $\sigma\equiv-1$ case; and on the class where the two channels are real and $\chi_1$ dominates in the cells where they oppose, that criterion is an \emph{equivalence} with $\beta\le3/16$, machine-checked.

The missing step is the \emph{sufficiency of that structure}: a proof that the seed-starved $d$-local read forces the two-channel bound. One of its two routes is closed, machine-checked: on a real-first-channel class the fan deficit is at most half the first-channel magnitude and that bound is attained, so at the necessary floor the aligned term reaches $5/16$ of the $5/8$ the criterion asks for; the fan is exactly the cosine of the character's phase against the view's, so what remains is a statement about that phase, and it is the $\chi_3$ channel's sign that has to make up the shortfall, which is how the five sub-threshold witnesses of Sec.~\ref{sec:dichotomy} pass. That bound is governed by the phase-aligned character and the second channel's sign rather than by the raw $\hat m_1$: explicit balanced $d=3$ circuits reach a high $\hat m_1$ yet $\beta>3/16$, because the residual character is phase-fanned. It is a known-hard lower-bound statement in the cyclic-group local-limit territory of \cite{KaneOstuniWu2024hamming,KaneOstuniWu2025symmetric,KaneOstuniWu2025uniform}.

Two further facts the measurements settle. Alignment is \emph{bought}, not inherited: over random $d=3$ samplers at each read size, the starved reads align but have no survivors, a small support tripping door (2d) by construction, while none of the survivors reaches its own aligned floor. And the shallowest non-linear onset within the balanced class is at depth $3$, every balanced depth-$\le2$ sampler's seed map being affine. Over the surviving, all-doors family the phase-aligned character and the across-fixing mean span are \emph{positively} correlated, so the two faces move together; whether the coupling is intrinsic is open, and it is the shared obstruction recorded at the end of Open Problem~\ref{open:dist}. The census, the constants of the two conditions, and the localization of the carrier are App.~C.
\end{open}
The two-door level form, residue $\cup$ rank, fails at small scale by the gated-column sampler of Sec.~\ref{sec:perfixer} and at every scale up to $n-1=4099$ by the side-disjoint family of App.~C, hence the third door; the three-door form is refuted as a dichotomy, App.~C, and the fourth, anti-concentration door (2d), closes the residual on the seed-limited class alone.

The residual splits into two forms. The Test-1 sliver $\beta\le3/16$ has its realized $\varepsilon$-far passers caught by (2d) on the shared stream of Sec.~\ref{sec:dichotomy}. The full premise $\beta\le1/4$ is where the known passers live: the small band's six carry $\beta>3/16$ and die at Test~1, as do the deeper tag-fixed ones, and the tag-varying passers populate the Test-1 sliver itself, App.~C.

The support side of the residual has been probed: lifting the $12$-seed sliver to $14$ seeds at $n-1=17$ holds all four doors clean while the support reaches $8457$, inside the shared stream's catch, App.~C, so the lift is a probe of the support curve rather than a residual witness.  The question is whether the support grows without bound with the seed count or saturates below every polynomial budget. The measured answer is that it grows with the seed count, though sub-exponentially, and falls with the target depth: at $n-1=18$ the exact support ascent reaches $3968$ of the $4096$ cube points at $12$ seeds, $11820$ in its second basin at $14$, and $12976$ at $16$, App.~C.

A directed search over the reachable $d$-local family sharpens the last clause: a $\beta$-minimizing search over exhaustive $3$-subset reads at $n-1\in\{10,12\}$ reaches $\beta\le3/16$ circuits only ever with either weight-law distortion, $\mathrm{histTV}\ge0.06$, or small support, $\le2^{10}$. The search finds \emph{zero} violations of this two-outcome disjunction over $437$ circuits, and the adversary maximizing the residual bias under $\{\mathrm{support}>2^{10},\mathrm{histTV}<0.06,\ \varepsilon$-far$\}$ finds that feasible region \emph{empty}. So the many-seed large-support $\varepsilon$-far imposter is not realized in the \emph{small}-$n$ reachable family, $n-1\in\{10,12\}$; the door combinations it does realize are mapped in App.~C.

The many-seed large-support, collision-loud imposter \emph{is} realized at larger depth, $n-1=16$, so its non-existence for all $\mathrm{NC}^0$ is false, and the question is which bounded-fan-in samplers realize it: a genuine lower bound rather than a missing simple lemma. The residual non-symmetric reach of $d$-local circuits is controlled by their Hamming-slice local-limit machinery \cite{KaneOstuniWu2024hamming,KaneOstuniWu2025symmetric} together with the $\Omega(\log(n/k)/\log\log(n/k))$ decision-depth bound of Filmus, Leigh, Riazanov, and Sokolov \cite{FilmusLeighRiazanovSokolov2023}, a known-hard lower-bound question that the literature states as further research on this tag, whose resolution would separate quantum from classical $\mathrm{NC}^0$ sampling.

In $\pm1$ form $L(z)=(-1)^{f(z)}$ with a $3$-bit read $z_S$, the least error over predictors reading $z_S$ is $\tfrac12\bigl(1-\mathrm{avg}\,\lvert\mathbb E[L\mid z_S]\rvert\bigr)$, attained by the per-cell majority read, and a sampler's own $\beta$ is at least that least error, so $\beta\le3/16$ forces $\mathrm{avg}\,\lvert\mathbb E[L\mid z_S]\rvert\ge5/8$, and $\mathrm{avg}$ is a sum of absolute values of linear functionals of the law, hence \emph{convex}. The structure theorem bounds $\mathrm{avg}$ on the symmetric-close class, the images close in total variation to a weight-only law. Each $\gamma'$-biased residual weight, a product law with per-coordinate bias the integer multiple $\gamma'$ of $2^{-d}$ the structure theorem supplies, is a wide binomial with a near-uniform residue, and on a symmetric near-Binomial image the conditional label mean depends only on the read's weight and stays, by the mod-$2p$ character rate, below the $5/8$ that $\beta\le3/16$ forces. On the residual class $\mathrm{avg}$ is a per-cell canceled signed character sum of the residual weight. The open question is therefore precisely a bound on two things: the per-cell mod-$2p$ asymmetry of that residual weight, and the cross-cell phase diversity a bounded read of at most $d$ seeds can produce. That is a genuine $d$-local lower-bound statement in the Kane--Ostuni--Wu and Filmus--Leigh--Riazanov--Sokolov locality-bound territory, whose $d$-stability holds on the reachable family but is not proved.

The reduction is complete up to the precise statement of this bound. Two structural facts about it are settled. First, the symmetric part is not the obstruction. Since the $\pm1$ label is a function of $|z|\bmod 2p$ alone, on a symmetric near-Binomial image the conditional mean of the label depends only on the read's weight and is bounded by the machine-checked mod-$2p$ character rate, so such images keep $\mathrm{avg}$ well below $5/8$. The residual is therefore the non-symmetric d-local reach. Second, a second-moment form bounds the target from one side. By Cauchy--Schwarz $\mathrm{avg}^2\le\mathbb E[\mathbb E[L\mid z_S]^2]$ is the $3$-bit predictability energy of the label, so $\beta\le3/16$ forces it above $25/64$, machine-checked, and the route closes exactly the samplers kept below the line, yet the full-cube baseline lies \emph{above} the line, so the bound is about $\le d$-seed reads, seed sharing, not the cube.

A seed-starvation picture sharpens this further. In the realized label-predictable $d$-local images few outputs read no read-seed, so almost all outputs are constrained by the bounded read. The resulting residual weight law, a rich mod-$2p$ family \cite{BeyersdorffDattaKrebsMahajanScharfenbergerFabianSreenivasaiahThomasVollmer2013}, is what must be bounded.

On the affine class the conjecture is proven with an exponential margin at any constant locality: Proposition~\ref{prop:affine} gives $\beta\ge1/2-2^{\frac d2-1}\sqrt p\,\cos(\pi/2p)^{n-1}$ for every full-rank affine sampler with more than $d$ odd-weight rows, and the complementary case of at most $d$ odd-weight rows falls in the parity-junta branch of Theorem~\ref{thm:perfixer}, bounding $\beta$ by the parity-leak floor $\tfrac12-\tfrac1{2p}+o(1)$ when the visible groups are many, the floor E4's measured $\beta$ sits just below. So the residual is confined to non-affine $z$-marginals, which would NC$^0$-sample the modular half-slice $H_c$ at full rank. One affine configuration sits outside both arguments, an affine coset of a hyperplane not through the origin, whose linear span is the whole space and whose linear part has rank $n-2$; the same character-sum computation reaches it, the parity character shifted by the coset offset, and bounds $\beta$ below by the parity-leak floor at every fixed locality, far above the dichotomy bound, so Test~1 catches it.

The symmetric case is confined by the Kane--Ostuni--Wu theorems. Their Hamming-slice sampling bound, Theorem~5.17 of the full version \cite{KaneOstuniWu2024hamming}, is stated additively, $\mathrm{TV}\ge1-\mathrm{loss}-(\tfrac12+\tfrac1{2p})$, with a loss term that is tower-scale, so that in that form the bound is vacuous at any experimental $n$. Its $n\to\infty$ content at the half-slice is exactly the constant this paper pins, $\tfrac12-\tfrac1{2p}$, which the even-weight reference sampler attains, so the contribution there is attainment rather than a new bound. Their sampling-structure theorem \cite{KaneOstuniWu2025symmetric} confines a $d$-local image that is $\tau'$-close to such a law to a mixture of the two parity slices and products with per-coordinate bias $\gamma'$, each component with residual bias far below the $5/8$ that $\beta\le3/16$ needs, quantified in Sec.~\ref{sec:dichotomy}.

What the confinement does not reach is the non-symmetric case, joint $(z,b)$ imposters whose conditionals within a half-slice are non-uniform, and it is there that Proposition~\ref{prop:affine}, with the parity-junta branch of Theorem~\ref{thm:perfixer}, closes the affine class. The details, the attainment constant and the fooling-distance reading of the two refuted families are in App.~C.
Within the constructions built here, only samplers passing all three doors at once challenge the repaired form, and frontal search finds none beyond the small band $n-1\le14$ of App.~C. That band's passers all die at Test~1; the $\beta\le3/16$ passers come from the tag-varying mutation search below rather than from this frontal one. An escalated directed search extends them to $n-1\in\{15,16,17,19\}$ at $\beta\le1/4$, the tag-fixed floor reaching $0.2209$, above $3/16$ by a $0.033$ margin, while a mutation search that also varies the tag reaches $0.1397$. The class is not bounded there: a larger-budget different-basin breach search reaches one at $n-1=22$, below the premise. Every breach-search passer is $\varepsilon$-far with $\beta>3/16$.

\noindent\textbf{Empirical phase transition.} A scan of $405$ random samplers whose output bits are XORs of $d\in\{1,2,3\}$ seed bits, at $42$ $z$-bits and $p=3$, shows a sharp bifurcation. Every random sampler has measured $\beta_{\rm pool}\in[0.487,0.495]$ at rank $23$--$42$, the only $\beta\approx0$ point is the degenerate block at rank $7$, and no sampler in between reaches full rank, Fig.~\ref{fig:influence}. The residue marginals corroborate the conjecture: no random sampler's residue exceeds total variation $0.011$ from uniform, so none concentrates on a half-slice, whose support would force variation at least $1/2$, and the residue-character modulus at full rank stays below $0.014$. These are plausibility checks on a restricted random model, not soundness evidence; read as a phase diagram, the empty corner, low violation at full rank, is the analogue of a forbidden region, and a completed Test~2 is the statement that the corner stays empty for every constant-locality sampler.

The \emph{non-degenerate} sub-case of Test~2 is disposed of by Watts--Parham's mod-$p$ anti-concentration, the bound on how much probability a weighted sum of bits can place on any single residue, together with the small-ball regime \cite{NguyenVu2013smallball}, the probability a sum of bits lands near a given value, of Test~2 inside their block model, in which the $z$-bits split into $t$ blocks whose per-block weights have nonzero increments. Their Fact~29 gives the decay $\sqrt p\,e^{-t/p^2}$ for any nonzero increments, their Corollary~30 turns it into $O(1/n)$-grade uniformity at $t=\Omega(p^3)$, and by their Lemma~31, whose coefficient hypothesis $0<a_i\le O(p/\log n)$ the model's increments satisfy, and Claim~32, every predictor agrees with the tag on at most $1/2+O(1/\log n)$ of inputs, so that $\beta\ge1/2-O(1/\log n)$, contradicting small $\beta$ \cite{BenewattsParham2026}. Any low-$\beta$ imposter falling under that model is therefore degenerate. Theorem~\ref{thm:perfixer} reproduces the same decay mechanism without the model hypothesis, by conditioning on a seed fixing. The open core is a universal structure test for degenerate NC$^0$ samplers.

The natural-proofs barrier \cite{RazborovRudich1997} does not apply, and the reason is precisely that the residual is a lower-bound question inside a constant-locality class. The criterion protects only a class hosting a strong pseudorandom-function generator, and what it would have to protect against here is a sampler, whose difficulty is not how hard it is to describe but what shape of distribution it can produce. Disguise is no help either: a fully pseudorandom $(z,b)$ puts the event $b\ne f(z)$ at probability $\beta=\tfrac12\pm\mathrm{negl}$, which Test~1 rejects outright, and $D_n$ itself is structured rather than pseudorandom, $\mathrm{AC}^0$ being learnable \cite{LinialMansourNisan1993}. The obstruction this paper meets is therefore not the natural-proofs one. It is the constraint a constant-locality read places on how a sampler distributes its randomness, which the collapse of Sec.~\ref{sec:perfixer} makes precise, and the residual of Open Problem~\ref{open:dist} is a statement inside that constraint.

The \emph{idealized} escape closes completely:

\begin{proposition}[Idealized g-varying escape is caught or vacuous]\label{prop:mixture}
Let $H_c:=\{z:|z|\bmod2p\in F^{-1}(c)\}$, and let $P_z=w\,\mathrm{Unif}(H_0)+(1{-}w)\,\mathrm{Unif}(H_1)$ be the $z$-marginal of an idealized $\beta\approx0$ sampler in which, per junta fixing, $z$ is uniform on the corresponding half-slice, and $g$-varying means the output bit $g$ selects the slice. Then
\begin{equation}
\begin{aligned}
\mathrm{TV}(P_z,\,U_{n-1})&=\bigl|w-\tfrac12\bigr|+o(1),\\
\mathrm{TV}\bigl(R_{P_z},R_U\bigr)&=\bigl|w-\tfrac12\bigr|+o(1),
\end{aligned}
\end{equation}
where $R_X$ denotes the $|z|\bmod2p$ marginal of the law $X$ and each $o(1)$ term is bounded by the L1 deviation $(2p{-}1)\cos(\pi/2p)^{n-1}$ of the character decay of Prop.~\ref{prop:decorr}. The mixture identities, the fiber-size estimates, and the residue-TV bound are verified in Lean.
\end{proposition}
\begin{proof}[One line; the identities are App.~D]
The two half-slices partition the cube, so the mixture density is $w/|H_0|$ on $H_0$ and $(1{-}w)/|H_1|$ on $H_1$; each fiber of $F$ contains exactly $p$ residues, and the residue of a uniform $z$ is within $\exp(-\Theta(n/p^2))$ of uniform on $\mathbb Z/2p$ by the character decay of Prop.~\ref{prop:decorr}, whence $|H_c|=2^{n-2}(1+o(1))$ and the mixture's total variation to uniform is $|w-\tfrac12|+o(1)$. The residue marginal puts mass $w$ on the $p$ residues of $F^{-1}(0)$ and $1{-}w$ on those of $F^{-1}(1)$, and the same character decay perturbs that by at most its own amount.
\end{proof}
Consequently an idealized $\beta\approx0$ imposter with $\mathrm{TV}(P_z,U_{n-1})\ge\varepsilon/2$ has residue total variation $\ge\varepsilon/2-o(1)$ and is caught by the residue test at any threshold in the window $(\gamma,\varepsilon/2)$, which lies below the forced value and above the device's $\gamma$-bounded deviation; the only exact escape, $w=\tfrac12$, yields $\mathrm{TV}(P_z,U_{n-1})=o(1)$ and is not an imposter at all. The formal residual of Open Problem~\ref{open:dist} is thus confined to non-idealized conditionals, asymmetric or non-uniform within a half-slice, where the numerics of Fig.~\ref{fig:influence} show that concentration forces rank deficiency.

\paragraph{What a completed verifier certifies.} A completed Test~2 would give a polynomial-sample, assumption-free classical certificate that the device's samples cannot be produced by any constant-locality classical simulation. The class the verifier is defined against is NC$^0$, constant locality per output bit with no bound on the seed length, while Watts--Parham rule out samplers of depth $o(\log\log n)$ restricted to $n+n^{\delta'}$ random bits, a statement their full version extends to $kN+N^{\delta'}$ i.i.d.\ Bernoulli bits of entropy $1/k$ per bit for their advice-free target (Theorem~54), so the exclusion is robust even against biased coins. The two classes are incomparable, ours weaker in depth and stronger in seeds, and the collapse of Sec.~\ref{sec:perfixer} holds at every seed count, its Tur\'an live-set bound only growing with that count; soundness against that same seed-unrestricted class is the open problem of Sec.~\ref{sec:open}. For slice-type marginals the seed-unrestricted lower bound is Kane--Ostuni--Wu's Theorem~5.17 of the full version \cite{KaneOstuniWu2024hamming}.

The collapse covers the $\Theta(\log\log n)$ boundary by a threshold adjustment; the extension is capped by the fact that $D_n$ itself is classically samplable to distance $O(1/n)$, Sec.~\ref{sec:open}.

\section{The per-fixing collapse theorem}\label{sec:perfixer}
The structural residual yields to a seed-fixing decomposition, the move Viola and Watts--Parham make for a $d$-local function but selected here by Tur\'an's theorem on the collision graph so that it reaches an arbitrary NC$^0$ sampler, followed by the verified character-sum machinery applied conditionally; this removes the non-degeneracy hypothesis of their bound.

\paragraph{Seed-fixing decomposition.} Let each $z$-output depend on at most $d$ seed bits. Call two seeds \emph{colliding} if some output depends on both; Tur\'an's theorem then yields an independent \emph{live set} $\Lambda$ of size $t=\Omega(\min\{\ell,\ell^2/(d^2n)\})$, where $\ell$ is the seed count and either branch is $\Omega(n/d^2)$ for $\ell\ge n$. Fixing all seeds outside $\Lambda$ to $\rho$, the $z$-bits split into independent single-seed groups $\{T_j\}_{j\in\Lambda}$, group $T_j$ being the outputs whose unique live seed is $j$, each contributing a weight taking two values $a_j(\rho)$ and $a_j(\rho)+\delta_j(\rho)$. Call group $j$ \emph{visible} if $\delta_j(\rho)\not\equiv0\pmod{2p}$, and let $m(\rho)$ count the visible groups; invisible groups shift the residue by a constant.

\begin{theorem}[Per-fixing collapse]\label{thm:perfixer}
Let $S$ be an NC$^0$ sampler with locality $d$ and $\beta(S)<3/8$, and let $p\ge3$ be odd. Call a visible group odd if its increment $\delta_j(\rho)$ is odd, and let $m_{\rm odd}(\rho)$ count the odd visible groups. For a $1-\tfrac83\,\beta(S)$ fraction of fixings $\rho$ of every live set $\Lambda$, where a live set is an independent set of the collision graph of size at least the Tur\'an bound of Sec.~\ref{sec:perfixer}, one of two collapses holds,
\begin{equation}
m(\rho)\;\le\;d+9p^2\ln\!\bigl(4\sqrt{2p}\bigr)/\pi^2
\qquad\text{or}\qquad
m_{\rm odd}(\rho)\;\le\;d,
\end{equation}
and, for every fixing, the conditional $z$-support lies in an $\mathbb F_2$-affine subspace of dimension at most $t$, and the conditional residue $|z|\bmod2p$ is a function of the $m(\rho)$ visible live seeds alone. In the second branch the residue parity $|z|\bmod 2$ is a function of the $m_{\rm odd}(\rho)$ odd groups alone, a $d$-junta.
\end{theorem}

\begin{proof}[Proof sketch; the machine-checked chain is App.~C]
Condition on $\rho$. The spectrum of $\Phi:=(-1)^F$ on $\mathbb Z/2p$ is supported on odd $k$, since $\Phi(r+p)=-\Phi(r)$ for odd $p$, and its parity coefficient is exactly $\hat c_p=1/p$. Expanding $\Phi(W)=\sum_k\hat c_k\,\omega^{kW}$ and factoring $\omega^{kW}=\prod_j\omega^{kW_j}$ over the independent groups, a $d$-local $g$ reading live seeds $K$ obeys
\[
\begin{aligned}
\bigl|\mathbb E[\Phi(W)\,(-1)^g\mid\rho]\bigr|&\le\tfrac{\mathbf 1[m_{\rm odd}(\rho)\le d]}p\\
&\qquad+\sqrt{2p}\,\cos(\pi/2p)^{m(\rho)-d},
\end{aligned}
\]
because the parity character $k=p$ is exactly one on visible groups of even increment and exactly zero on odd-increment groups, so it survives with its coefficient $1/p$ only when every odd visible group lies in $K$, while every other character of $\Phi$ has $k$ invertible modulo $2p$ and each visible group contributes at most $\cos(\pi/2p)$, at most $d$ groups being exempt through $K$. The affine clause follows because after the fixing the outputs outside the live groups are constant and each group toggles between two fixed patterns, so the conditional $z$-support is a translate of the span of at most $t$ toggle patterns, and the conditional residue depends on the live seeds only through the visible groups. Summing, $\beta(\rho)\ge\tfrac12-\tfrac12\bigl(\tfrac{\mathbf 1[m_{\rm odd}\le d]}p+\sqrt{2p}\cos(\pi/2p)^{m(\rho)-d}\bigr)$, and since $-\ln\cos(\pi/2p)\ge\pi^2/(9p^2)$ with slack that keeps the printed threshold deliberately loose, fixings with $m(\rho)>d+9p^2\ln(4\sqrt{2p})/\pi^2$, where the cosine term is at most $1/4$, \emph{and} $m_{\rm odd}(\rho)>d$, where the parity term vanishes, have $\beta(\rho)\ge\tfrac38$, so they have measure at most $\tfrac83\,\beta(S)$.
\end{proof}

At the numerics' scale both branches bind: at $p=3$, $d=2$ the first branch's threshold $d+9p^2\ln(4\sqrt{2p})/\pi^2\approx21$ sits well inside the $n-1\le112$ $z$-bits. Thus a $\beta$-small sampler must, on a $1-\tfrac83\beta(S)$ fraction of fixings, either reduce the residue-relevant randomness to $d+Cp^2\log p$ live bits, or reduce the residue parity to a $d$-junta while the residue moves only in even steps plus those $d$ toggles.

No poly-sample test observes either branch directly: the visible-group count is internal to the construction, and certifying junta structure from samples is exponentially costly, Sec.~\ref{sec:open}. The theorem's role is to pin the structure a completed Test~2 must catch; its one verifier-facing consequence, the conditional soundness extension to depth $\le\log\log n$, is stated in Sec.~\ref{sec:open}.

The hypothesis $\beta(S)<3/8$ is the dichotomy level itself: it covers the Test-1 survivors, $\beta<3/16$, and the conjecture class, $\beta\le1/4$, with room, the measure bound being the Markov ratio $\beta(S)/(3/8)$.

The $2p$-block sampler collapses completely, $m\equiv0$; the dyadic terminal construction E4, one of the four adversarial constructions of Table~S1, lies in the second branch, its increments $\{1,2,4\}$ giving $m_{\rm odd}=1\le d=2$, which is why its measured $\beta_{\rm pool}=0.328$ lies just below the asymptotic parity-leak value $\tfrac12(1-\tfrac1p)=\tfrac13$, the exhaustive best-2-local value up to sampling noise, the full branch bound being vacuous at E4's visible-group count.

Watts--Parham's non-degenerate regime is the opposite extreme, $m$ maximal, forcing $\beta\approx1/2$, which reproduces their Fact-29 decay mechanism as the $m$-maximal case of Theorem~\ref{thm:perfixer}.

Fig.~\ref{fig:cascade} confirms the knee numerically; the knee construction uses unit, hence odd, increments, so its applicable bound is $\tfrac12-\tfrac12\sqrt{2p}\cos(\pi/2p)^{m-d}$ without the parity term.

The escape battery shows the residual mixture case: E2, a fresh random mask on every position, passes both the residue and rank doors and is caught by Test~1; the odd-weight, weight-degenerate families of App.~C join it through both structure doors. For these affine imposters the Test-1 catch is provable with an exponential margin. Excluding the six small-band coexistence points that Test~1 alone catches, every imposter construction in Table~S1 with $\beta<0.4$ is rank-deficient, residue-concentrated, or weight-distorted, the gated-column escape of Sec.~\ref{sec:perfixer} falling in the third class.

\begin{proposition}[Full-rank affine imposters pay $1/2-o(1)$, at any constant locality]\label{prop:affine}
Let $P$ have an affine $z$-marginal, $z=S\cdot V\oplus z_0$ with $S$ uniform on $\{0,1\}^K$ and $V$ of rank $n-1$, let $\mathcal M$ be the set of odd-weight rows of $V$ with $|\mathcal M|>d$, and let $\hat g_T$ be the Walsh coefficients of a $d$-local output bit $g$ on the $d$ seed coordinates it reads, its Fourier coefficients over those seeds. Then
\begin{equation}
\begin{aligned}
\beta=\Pr\bigl[g(S)\ne F(|z|\bmod 2p)\bigr]\ \ge\ \frac12-\frac12\Bigl(\sum_{|T|\le d}|\hat g_T|\Bigr)\\
\times\sqrt p\,\cos(\pi/2p)^{\,n-1}\ \ge\ \frac12-2^{\frac d2-1}\sqrt p\,\cos(\pi/2p)^{\,n-1},
\end{aligned}
\end{equation}
the second step by Cauchy--Schwarz on the $2^d$ coefficients. A single-character $g$ has $\sum_T|\hat g_T|=1$ and recovers $\beta\ge\tfrac12-\tfrac12\sqrt p\,\cos(\pi/2p)^{n-1}$, the case machine-checked in Sec.~\ref{sec:lean}; at $d=2$, the locality of the battery and the scans, the envelope prefactor is $1$.
\end{proposition}
\begin{proof}[One line; the character computation is App.~D]
The parity factor is exactly linear in the seeds, $(-1)^{|z|}=\chi_{\mathcal M}(S)$ with $\chi_{\mathcal M}=(-1)^{\sum_{j\in\mathcal M}s_j}$, so the correlation reduces to a shifted character sum whose modulus is at most $\cos(\pi/2p)^{n-1}$ once $|\mathcal M|>d$, and a Cauchy--Schwarz step over the $2^d$ Walsh coefficients of a general $d$-local $g$ gives the envelope. The dichotomy is the affine twin of Theorem~\ref{thm:perfixer}: the kernel plays the live set, and the collapse is exact rather than measure-most.
\end{proof}

Linearity extends the bound to non-affine mixtures: if $P_z=w\,P_1+(1{-}w)\,P_2$ each $P_i$ a full-rank affine marginal meeting the hypotheses, then $|\mathbb E_P[\Phi\,\chi_T]|\le w\,B+(1{-}w)\,B=B$ for $B$ the proof's bound, so no mixture passes both structure doors below $\beta=1/2-o(1)$ either. Proposition~\ref{prop:affine} covers E2, whose mask rows alone span $\mathbb F_2^{n-1}$, and the full-rank odd-weight families: at the battery's $112$ $z$-bits it gives $\beta\ge1/2-1.7\cdot10^{-7}$, at the scan's $42$ $z$-bits $\beta\ge0.4959$, the exponent $n-1$ there being the $42$-bit $z$-dimension, with the debiased re-measurement $\beta\in[0.4961,0.4973]$ at $M=2\cdot10^5$ samples sitting above the $n=42$ bound. The $M=600$ readings $0.450$ and $0.415$ are the selection bias of taking the minimum over the $172$-predictor pool and over the samplers drawn at each cell. From the proven per-predictor base $0.4959$ the floor is $0.4959-\sqrt{2\ln172}/(2\sqrt{600})\approx0.43$ over the predictors alone and $\approx0.42$ once the sampler draws are included, and at $M=2\cdot10^5$ samples the residual floor from the unshifted base $0.5$ is $0.4964$, inside the re-measured range.

\begin{figure}[!htbp]
\centering\includegraphics[width=\linewidth]{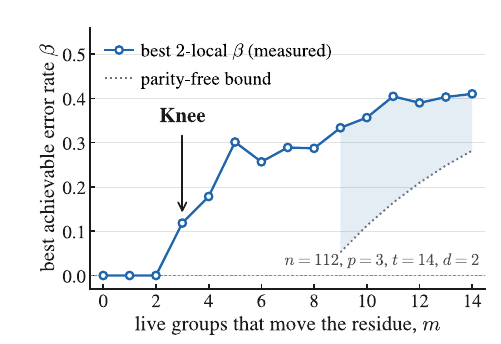}
\caption{The collapse theorem, measured. The horizontal axis is $m$, the number of live groups whose increment moves the residue $|z|\bmod 2p$, a live group being the set of output bits toggled by one free seed bit and its increment the weight change that toggle causes; the vertical axis is the best violation rate $\beta=\Pr[b\ne f(z)]$ the cascade construction reaches. The sweep is at $112$ $z$-bits, $p=3$, $t=14$ independent single-seed groups and locality $d=2$. Measured $\beta=0$ while $m\le d$, rising to $0.118$ at $m=3$, the knee. Dotted: the general parity-free smoothing bound $\tfrac12-\tfrac12\sqrt{2p}\cos(\pi/2p)^{m-d}$, drawn only where it is positive, which at these parameters means from $m=9$; below that it says nothing. The theorem's per-fixing form carries the exponent $m$ and is the tight one.}
\label{fig:cascade}
\end{figure}

\paragraph{The cascade: terminal dichotomy and the named residual.} Theorem~\ref{thm:perfixer} reduces the residue, conditionally, to few live bits plus the fixed part's contribution $|z_{\rm fix}(\rho)|\bmod2p$. At the terminal end of the recursion, with no fixed part, the dichotomy closes:

\begin{proposition}[Terminal dichotomy]\label{prop:terminal}
Suppose the residue takes the terminal form $W=(c+\sum_{j\in V}\delta_jr_j)\bmod2p$ for independent live bits $r_j$, and let $h:=\#\{j:\delta_j\not\equiv0\pmod p\}$ count the groups that move the \emph{mod-$p$} residue; note $h\le m(\rho)$, since a group whose increment is a multiple of $p$ but not of $2p$ counts in $m$ and not in $h$, being invisible modulo $p$ yet visible modulo $2p$. Then \textup{(i)} the mod-$p$ marginal of $W$ has support of size at most $2^h$, so its total variation from uniform is at least $1-2^h/p$, nonvacuous when $2^h<p$, with equality when the marginal is uniform on a support of size $s$, the case verified. For $h\le d$ with $2^d<p$ the bound tends to $1$ and Test~2a catches. \textup{(ii)} For $h$ at or beyond the collapse threshold $d+9p^2\ln(4\sqrt{2p})/\pi^2$, every character of $\Phi$ except the parity character decays on the visible groups, so $\beta\ge\tfrac38$ at the threshold, and $\ge\tfrac38-\tfrac1{2p}$ when the parity character survives with its coefficient $1/p$. The bound improves toward the parity-leak floor $\tfrac12-\tfrac1{2p}$ as $h$ grows, and at the operating scale $p=3$ it clears the dichotomy bound $\varepsilon/2=3/16$, since $\tfrac38-\tfrac16=\tfrac5{24}$.
\end{proposition}
\begin{proof}[One line; the cases are App.~D]
For (i), each live bit contributes $\delta_j$ or $0$, so the support of the mod-$p$ marginal is contained in $c$ plus the subset sums of the $h$ nonzero increments, at most $2^h$ points, and any distribution supported on $s$ points is at distance at least $1-s/p$ from uniform. Part (ii) is the $K=\varnothing$ case of the argument of Theorem~\ref{thm:perfixer}. Between the two the toggle structure rules: a group with $\delta_j=p$ flips $F$ exactly, while even-increment groups leave the parity invariant, so the parity character sees only the odd groups and the step spectrum sees all of them.
\end{proof}

The sharpest terminal escape is construction E4 of Table~S1, the dyadic one: dyadic visible weights $\delta_j=2^j$, $j<\kappa$, $\kappa=\lceil\log_2 2p\rceil$, followed by equalizing padding blocks that spread the residue over all $2p$ residues. The dyadic groups alone would leave total variation $r(2p-r)/(2p\,2^{\kappa})$, with $r=2^{\kappa}\bmod 2p$, which is $1/6$ at $p=3$; the padding closes the gap, and E4 passes Test~2a with empirical residue total variation $0.010$ at the battery's sample size.

Its best-$2$-local error measures $\beta=0.328$, by exhaustive search over its $19$ seeds up to sampling, and it is rank-deficient, rank $20$ against $112$ $z$-bits, so Test~2b catches it. At the operating point $\varepsilon=3/8$ of Sec.~\ref{sec:prelim}, Test~1 catches it as well, since $0.328>3/16$, leaving E4 the tightness check of the parity-leak term rather than a Test-2b-only catch. The named residual of the completeness question is thus a single \emph{mixture-transfer} step. Theorem~\ref{thm:perfixer} leaves the fixed-part centers $c(\rho)$; if the fixed part regenerates visible structure it is itself subject to the same collapse, recursively. Closing the recursion means proving that either the residue becomes a function of $O(d)$ seeds, caught by Test~2a as in (i), or the support degenerates, caught by Test~2b by Lemma~\ref{lem:halfspan}; that is exactly the statement that no mixture of terminal dichotomies across fixings passes all three doors. E2-type escapes, full-rank masks paying $\beta=1/2-o(1)$ by Proposition~\ref{prop:affine}, show the pressure at every scale.

\paragraph{The mixture-transfer mechanism: character cascade.}
The named residual admits a sharp formulation. Expanding $(-1)^{F(W)}$ over the spectrum of Fig.~\ref{fig:mix}, support on odd $k$, $\sum_k|\hat c_k|=\Theta(\ln p)$, $|\hat c_p|=1/p$ exactly, and conditioning on the fixing $\rho$ gives the \emph{character cascade}
\begin{equation}
\varphi_k \;=\; \mathbb E_\rho\Bigl[\psi^{\rm fix}_k(\rho)\cdot\prod_{j\in\Lambda}c_{j,k}(\rho)\Bigr],
\end{equation}
where $\varphi_k=\mathbb E\,\omega^{kW}$ is the global residue character, $\psi^{\rm fix}_k(\rho)=\omega^{k|z_{\rm fix}(\rho)|}$ the fixed part's phase, and $c_{j,k}\le\cos(\pi/2p)$ the verified per-group decay on visible groups. Closing the recursion is the assertion that no mixture beyond the terminal-dichotomy outcomes is stable under this cascade; that assertion is exactly Open Problem~\ref{open:view}.

\begin{figure}[!htbp]
\centering\includegraphics[width=\linewidth]{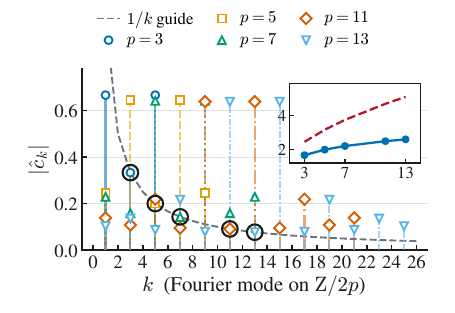}
\caption{What the collapse theorem prices. Main panel: the Fourier coefficients $|\hat c_k|$ of the tag's sign $(-1)^F$ at mode $k$ on $\mathbb Z/2p$, one curve per prime $p=3,5,7,11,13$, with $F$ the tag of Fig.~\ref{fig:residue}. The support lies entirely on odd $k$, and the coefficient at $k=p$ is exactly $1/p$, circled, against the dashed $1/k$ guide: that coefficient is the parity leak, the one term of the spectrum that survives when every other character decays. Inset: the sum of the coefficients against $p$, dashed at the machine-checked bound $\sqrt{2p}$, which is loose since the growth is logarithmic in $p$. The parity-aligned even-weight example of Sec.~\ref{sec:dichotomy} is measured in Fig.~S1.}
\label{fig:mix}
\end{figure}

Its mechanism is the \emph{rank--residue linkage}: full $\mathbb F_2$-rank of the $z$-support requires the fixed-part bases to carry at least $n-t$ dimensions of variation, since group toggles contribute at most $t$. Base variation enters the residue through the base weight, so the centers $c(\rho)$ become a function of unboundedly many fixed seeds, which a $d$-local $g$ cannot track. The perturbation step of this argument is the machine-checked triangle inequality.

\paragraph{The terminal frontier: the rank door is essential.} The residual admits an exact computation that isolates the rank door. Restrict to the \emph{terminal class}: a per-fixing residue with $\rho$-independent increments and an arbitrary center law, exactly enumerable at small scale. Every \emph{realization} of the abstract class's concentrated points is caught, a one-parameter dial over twenty-eight full samplers landing every point at $\beta\in[0.486,0.491]$. The structural residual is thereby pinned to a concentrated center law at full $\mathbb F_2$-rank; the analytic skeleton and the enumeration are App.~C.

\paragraph{The non-terminal residual: a gated escape.} The $\rho$-adaptive mechanism is realizable at $d=2$, and it escapes the residue and rank doors at small scale. The \emph{gated column} reads a pair so that the per-fixing residue carries $\rho$-dependent increments over five live seeds, at $n-1=11$ columns each $2$-local. The resulting sampler is an exact counterexample to the level form of Open Problem~\ref{open:dist} at $p=3$, with $\beta=0.1719<3/16$ the exact optimum over the full predictor pool and rank $11$, and it dies at the weight-histogram door, which subsumes the residue door. The construction and the battery that confirms the linkage are App.~C.

\paragraph{The regenerating-structure class is Test-2a-safe, provably.} The sharpest non-symmetric escape is a sampler whose visible structure \emph{regenerates} per fixing, each of $t_g$ disjoint groups reading a live seed through an independent balanced $3$-local function. The whole chain is machine-checked in the Regenerating module, so such a sampler provably passes Test~2a at full rank, which is why the residual is not the regenerating family. The construction and the chain are in App.~C.

\paragraph{Fiber rigidity at $\beta=0$.}
The mixture-transfer statement reduces at $\beta=0$ to a fiber-rigidity question. If $g=F\circ R$ exactly, then $F(R)$ is a $d$-junta reading seeds $K$, and within each $K$-slice the residue $R$ can only navigate its fiber: every reachable value lies in $F^{-1}(g(\sigma))$. Two verified rigidity facts govern the navigation:

\begin{lemma}[Fiber rigidity]\label{lem:fiber}
For odd prime $p$ and any $r$: \textup{(i)} the $p$-shift always leaves the fiber, $F(r+p)=1-F(r)$, verified; \textup{(ii)} the even coset sweeps both $F$-values, some even shift $2k$ with $k\le p$ exits any prescribed fiber, verified, via $2^{-1}=(p{+}1)/2$ modulo $p$ targeting both residues $0$ and $(p{+}1)/2>p/2$.
\end{lemma}

Consequently a subgroup coset never fits inside a fiber, and within-slice navigation must use genuinely position-dependent steps: constant per-group deltas can navigate, the pair $\delta=(5,5)$ sweeps the full fiber $\{0,4,5\}$ of $p=3$, but the navigation set is confined to a single fiber, of $p$ points.

\begin{remark}[$\beta=0$ dichotomy, informal]\label{rem:beta0}
At $\beta=0$ the picture is the following, supported by Lemma~\ref{lem:fiber} and the battery numerics of Table~S1 but not yet a proof. Either the slice fibers coincide across junta fixings, in which case the residue marginal is supported on one fiber and has total variation at least $1/2$ from uniform, caught by Test~2a. Or distinct slices use distinct fibers, which forces $g$ to read the slice variable and the $z$-support to vary its base weight across slices, so the bases carry the rank. In the constant-increment regime, where every slice carries the same increment, the reachable set of each slice is a fixed subset-sum set $S(\delta)$ that must fit in every shifted fiber; no nontrivial coset does, by Lemma~\ref{lem:fiber}.

The step completing this dichotomy is the rank bookkeeping for weight-degenerate affine families, bases whose base weight is constant modulo $2p$ while spanning $\mathbb F_2^{n-1}$. The all-even-weight case is machine-checked in the WDRank module of Sec.~\ref{sec:lean}, which pins three things: the parity lock, every even-weight combination having even weight; the residue door, the residue locked to the $p$ even residues with total variation $\ge1/2$ against uniform; and the rank bookkeeping, the rows spanning at most $N-1$ of the $n-1$ $z$-bits via the parity-kernel dimension. Such a family is therefore caught by the residue door or the rank door. The transfer across fixings that Open Problem~\ref{open:view} poses remains open.
\end{remark}
Every $\beta=0$ construction built in this work within the remark's model, the $2p$-block, the escape, E1, E3, and the $\delta{=}(5,5)$ fiber navigator, is residue-confined with total variation $\ge1/2$ or rank-deficient, as the remark predicts. The four-door construction of Sec.~\ref{sec:dichotomy} is $\beta=0$ outside that model, full rank and residue-uniform, and falls at door (2e) instead.

The residual is thereby two independent knobs rather than one. A direct measurement over the full-rank witnesses that clear both structure doors, App.~C, shows the label is \emph{not} per-fixing-determinate, its purity, the majority answer's rate within a fixing, well above the fully-ambiguous floor, and the per-fixing $\beta$ uniform in mean. So $\beta$ is a global tag$\leftrightarrow$label correlation, while the residue and weight-histogram doors are mediated by the per-fixing concentration and the across-fixing mean span, whose geometric core is machine-checked.

A census of the whole $\beta\le1/4$-reachable family, App.~C, confirms this from the opposite side and quantifies its cost. $\beta\le1/4$ is attained at span below $2$, so span is \emph{not} what buys $\beta\le1/4$; every four-door passer falls inside the bounded span window $[3.0,4.0]$, whose ends no passer crosses; the higher-span deep instances are weight-distorted and caught; and every $\beta\le1/4$ instance of small span trips both the residue and the weight-histogram door with wide margin. The machine-checked bound $\mathrm{span}\ge2\sqrt{V-v}$, with $V$ the marginal variance and $v$ the mean conditional variance, holds across the whole reachable family.

The across-fixing mean span is therefore itself a structured covariate of the doors. A re-read of the Kane--Ostuni--Wu proof at the level of its stages finds no object controlling an across-fixing span, so bounding it is new mathematics rather than a bounded piece of the existing machinery, and bounding the span would bound the Test-2 doors without subsuming $\beta$. The measurement that closes the naive route is that full rank does not force increments equal to one, App.~C.

\paragraph{The mixture-transfer face of Open Problem~\ref{open:view}.} What that object asks, on the mixture side, is a proof that the three doors, residue uniformity, full $\mathbb F_2$-rank and weight-near-Binomial, close the $\varepsilon$-close class: every $\mathrm{TV}(P_z,U_{n-1})\le\varepsilon/2$ NC$^0$ imposter at $\varepsilon$-far from $D_n$ either has $\beta\ge\varepsilon/2$, which Test~1 rejects, or trips one of them. Soundness then reduces to Open Problem~\ref{open:dist}.

The natural ``beyond-band'' version, that no full-rank $\beta\le1/4$ all-three-door passer exists past the small coexistence depths, is \emph{false}: a larger-budget, different-basin search finds one at $n-1=22$, near-uniform, full rank and near-Binomial, so that set is not bounded in depth. Its members are always $\varepsilon$-far and always carry $\beta>3/16$, so Test~1 rejects each and soundness is not threatened; what remains open is therefore the \emph{sufficiency} of the doors on the $\varepsilon$-close class, and the distinct, narrower $\beta\le3/16$ sliver. The tempting reformulation, that full rank forces the fixed-part centers to carry a seed-complexity linear in $n$ and hence $\beta\ge\varepsilon/2$, is \emph{not} an equivalence: the in-band passers are full rank and pass all three doors at $\beta\le\varepsilon/2$, so the refuted reading is the one that would let the rank door stand in for Test~1. The structural reason is that the $z$-support rank is a property of the \emph{fixed columns} while the across-fixing center mixture is a property of the \emph{fixing}, so the two are decoupled rather than linked.

The criterion's chain is machine-checked in full, and it rests on two premises: a conditional-variance budget below the model's resolution, and a covering window of width mean span plus $\ell^1$ increment mass. Across the scanned family the window exceeds the mean span on all $273$ census rows and $14$ directed ones, by a median factor of $3.7$, and the realized dispersion sits \emph{within} the fixings rather than across them, which is why the chain clears $\tau'$ at the model's budget and not at the realized one. What the chain cannot recover from the doors is the conditional variance itself; the surplus over the premise is exactly the sum of two endpoint deficits, and it vanishes when the extreme conditional means sit on the widest fixings. The chain, its floors, the endpoint decomposition and the tag-cost mean's split are App.~C.

\paragraph{The variance face of Open Problem~\ref{open:view}.} Toward the transfer statement the \emph{weight} side of the $\beta\le1/4$ premise admits a genuine, optimization-derived criterion: for a weight law $X$ on $\{0,\dots,N\}$ with mean $N/2$, a variance deficit $\delta^*(N)$ below $N/4$ is what a weight-histogram-clean sampler would need, and the criterion is the statement that no such sampler exists. It does \emph{not} follow from the doors, and closing it as a theorem would not by itself prove a beyond-band item, none holding, since the near-uniform, full-rank, near-binomial passers that do recur at larger depth all carry $\beta>3/16$ and so are caught by Test~1. A direct measurement ties the criterion to the open core all three faces share: by the law of total variance the across-fixing mean-spread term \emph{is} the across-fixing mean span, and the rank--span coupling is window-local rather than global. The criterion, its optimization and the measurements are App.~C. The nearest classical trade-off between a rank parameter and a concentration does not apply here: Kaufman and Lovett's theorem states that a degree-$d$ polynomial of bias at least $\varepsilon$ has rank at most $c(d,\varepsilon)$ \cite{KaufmanLovett2008}, but their rank is the Green--Tao decomposition rank of a single polynomial and their constant-degree hypothesis fails for $f$, whose degree over $\mathbb F_2$ is $\Theta(n)$.

\section{Lean formalization of the analytic core}\label{sec:lean}
The analytic statements carry machine-checked counterparts: the structural lemma, the $2p$-block identity, the rank-test backbone, the idealized-mixture closure, the residue-rigidity congruence of Eq.~\eqref{eq:rigidity}, and the best-predictor identity that makes Test~1's rejection level and Test~2's floor one statement, together with the four links back to the sampler.

The formalization is in Lean~4 against Mathlib and builds with zero errors, zero \texttt{sorry}, Lean's marker for an unproved statement, and only Lean's three standard axioms, propositional extensionality, choice and quotient extensionality. Every statement machine-checked here is stated and proved, or sketched, in the body.

The collapse theorem is machine-checked chain by chain with the translation of an arbitrary NC$^0$ sampler into the fixing model it assumes, Test~1's estimation budget, Proposition~\ref{prop:affine}, and the residue-subsumption inequality with the $p=3$ smoothing bounds that pin $m_0\le18$ unconditionally at $\beta\le1/4$, with $m_0$ the visible-group count, and their every-prime analogues. What is deliberately \emph{not} formalized is the structural step of Open Problem~\ref{open:dist} itself and the sufficiency of the structure test's doors for the full class. The per-theorem mapping is App.~C.

\section{What is delivered, and what remains}\label{sec:open}
Verification is the credibility infrastructure of a sampling-advantage claim, and this paper builds the check for the sharpest shallow-circuit separation known before reporting where that check ends. The fifth door reads the constant-locality class's correlations rather than a distance. It is the first of four unconditional deliverables, the low-degree member of the structure test's five doors (2a)--(2e), with the construction that the four symmetric doors cannot reach. The other three are a test that assumes nothing about the imposter, Test~1, sample-optimal in its own tolerance $\delta$ up to the logarithmic factor; the machine-checked collapse theorem; and the hub rigidity that makes the fifth door sound on the bounded pinned-residue class.

\paragraph{Status of the verifier.} Test~1 catches every imposter with $\beta$ clearing the dichotomy's bound $\varepsilon/2$ at $O(\delta^{-2}\log(1/\eta))$ samples, its rejection level placed one tolerance below that bound so that no estimation band is left open. No single door catches every sampler built here, and the binding door differs by family. Door (2d) catches the seed-limited ones once its budget reaches the support it must cover. The gated-column family of Sec.~\ref{sec:perfixer}, the small-scale escape from the bare residue-plus-rank union, dies at the third, the weight-histogram door. The affine, regenerating and symmetric classes are caught by Test~1 rather than by Test~2. The three tiers the paper keeps apart are these: proved results, measured ones under a stated model and budget, and predicted ones, which are targets rather than demonstrated performance and are the only incomplete tier.

The affine case is closed at every constant locality by Proposition~\ref{prop:affine} together with the parity-junta branch of Theorem~\ref{thm:perfixer}, and the symmetric case is confined by the Kane--Ostuni--Wu theorems. Theorem~\ref{thm:perfixer} confines the surviving class to few residue-relevant bits or a parity junta on a $1-\tfrac83\beta$ fraction of fixings, at least half of the fixings at the Test-1 survivor premise $\beta\le\varepsilon/2-2\delta=25/144$. The step bounding that fraction is machine-checked in Lean.

Soundness of Test~2 as a whole is the single structural residual, Open Problem~\ref{open:dist}, now located precisely: it is a sampler satisfying
\begin{equation}\label{eq:residual}
\begin{aligned}
&\mathrm{TV}\bigl(P,D_n\bigr)\ge\varepsilon,\qquad\beta(P)\le3/16,\\
&\text{all five doors clean,}\qquad
\mathrm{TV}\bigl(P_z,P_z^{\mathrm{sym}}\bigr)\ge\varepsilon/2-\tau',\\
&1/\textstyle\sum_z P_z(z)^2>\text{every budget},
\end{aligned}
\end{equation}
Eq.~\eqref{eq:residual} is the residual; the last line is the collision-effective support that door~(2d) prices. The symmetry line is proven here, and it closes the route the other four invite: since the symmetrized $z$-marginal's distance from uniform is exactly the weight-histogram reading, the triangle inequality gives the displayed floor, and the bound moves the wrong way for that route, so no tightening of the five doors can bring the residual into a class whose hypothesis is closeness to symmetric, which is the shape the structure theorems need.

What is delivered is a reduction whose content is sharper than the separation. The least error over predictors reading the string half alone is at most the sampler's own $\beta$, so $\beta\le3/16$ forces an average magnitude of the label's conditional mean of at least $5/8$. A balanced label independent of the string half is unpredictable, so $\beta\le3/16$ forces the label's read-set to meet the string half's seeds, and a seed that no string-half coordinate reads makes the label's signed sum vanish on every fiber. The argument then needs a statement that those shared seeds are read by enough string-half coordinates to be \emph{determined}, which is a statement about the read's \emph{wiring}: constant locality bounds how many seeds one output reads, and not how many outputs one seed feeds. The predictability-energy floor and the fan-out lower bound are both machine-checked.

That residual is not closed here: testing $\mathrm{TV}(P_z,U_{n-1})$ against uniform over a $2^{n-1}$-symbol alphabet is identity testing at the exponential $\Theta(2^{(n-1)/2}/\varepsilon^{2})$ \cite{Canonne2022}. The fooling-construction framework of Bun, Desai, and Ferreira Pinto \cite{BunDesaiPinto2026} is the nearest such machinery, but it covers bounded \emph{classes of distinguishers} under fooling distance, not the $\mathrm{NC}^0$-samplable family itself, so Open Problem~\ref{open:dist} is best read as a sample-complexity barrier, not an impossibility. A $\mathrm{poly}(n)$-sample verifier is therefore confined, on the tested family, to the $\beta$-case together with $z$-structure proxies.

A $d$-local image that is $\tau'$-close to a symmetric distribution is close to a mixture of the parity slices and $\gamma'$-biased products, by the Kane--Ostuni--Wu structure theorem \cite{KaneOstuniWu2025symmetric}, and each such component has residual bias far below the $5/8$ that $\beta\le3/16$ needs. The weight-histogram door tests a necessary proxy for that theorem's hypothesis rather than the hypothesis itself, so the confinement is a statement about the symmetric-close images and not about every weight-law-clean passer: the residual is precisely the non-symmetric $d$-local reach the theorem cannot exclude, and the comparison is asymptotic in $\tau'$.

The items below carry the completeness gap, the two named open problems with their affine case closed and their symmetric case confined, and extend the framework to the advice-free function and to other shallow-circuit constructions.

\textbf{Residue, rank, and weight-law completeness.} The three-door form is refuted, App.~C, so the residual is the $z$-marginal separation $\mathrm{TV}(P_z,U_{n-1})$ of Open Problem~\ref{open:dist}, and its $\mathrm{poly}(n)$-sample testability is open rather than an information-theoretic impossibility. Open Problem~\ref{open:view} stands beside the residual rather than above it: settling it would fix the weight law of every low-$\beta$ sampler, which is what a completed structure test must read, and unsettled it leaves the two branches of the dichotomy open together. Rank alone is insufficient, the tiny-support sampler $\{0,e_i\}$ having $\beta=0$ and full rank while the residue door catches it.

The two analytic cores are verified, Prop.~\ref{prop:decorr} and Lemma~\ref{lem:halfspan}. What the residual still leaves open is the structural step, the non-symmetric realizability of a concentrated center law at full rank, the case of joint $(z,b)$ imposters with non-uniform conditionals inside a periodic half-slice, adjacent to the Kane--Ostuni--Wu theorems \cite{KaneOstuniWu2024hamming,KaneOstuniWu2025symmetric,KaneOstuniWu2025uniform}. The symmetric case beside it is sharpened: the parity-aligned even-weight attainment example's residue distance is machine-checked exact at uniform-on-evens, and its half-slice constant is confirmed exactly at $p\in\{3,5,7,11,13\}$ across $n$, App.~C.

At small scale the level form is settled \emph{negatively}: the random $2$-local family returns samplers passing all three repaired doors at once, but every one carries $\beta>3/16$, so Test~1 alone catches each, and no sampler in the combined trials passes both tests. The repaired doors are therefore not individually a dichotomy on $\beta\le1/4$, and the completeness question at the verifier-relevant premise $\beta\le3/16$ remains the narrow band identified above. The residual localizes to the per-fixing collapsed corner, and the pure group form of the collapse yields a machine-checked architectural band: $\mathbb F_2$-independent columns force $n-1\le m_0+2$, with $m_0$ the visible-group count, so at $p=3$ under $\beta\le1/4$ a four-door passer beyond $n-1=20$ cannot be written in group form at all. Open Problem~\ref{open:view} is thereby sharpened to its final form: what is missing is exactly the combination of rank-realizable columns with an across-fixing center mixture. The per-depth scans and the full census are App.~C.

Affine door-passers are provably at $\beta=1/2-o(1)$ by Proposition~\ref{prop:affine}, and the linear rank-deficient case is closed by an exact identity, the residue-character sum over the code's row space being the subset-sum character sum of the kernel code's weight distribution, governed by that distribution rather than by the minimum distance alone \cite{Kleitman1966,NguyenVu2013smallball}. The threshold question therefore lives entirely in the non-affine class.

\textbf{The advice-free function: the dichotomy transfers.} Watts--Parham's advice-free sampler outputs $f=\mathrm{pmmajmod}_p$, a majority-modulo-$p$ read of the tree statistic their sampler computes rather than of the weight, so it is not a function of the weight \cite{BenewattsParham2026}, and the dichotomy reduction transfers verbatim: Test~1 and Test~2b apply unchanged, and Test~2a becomes a uniformity test of their tree statistic modulo $p$ on $p$ symbols \cite{Paninski2008}. The resulting dichotomy is established up to the boundary-strip anti-concentration of their Lemma~31, whose flip-set combinatorics and probability estimates are machine-checked here. The advice-free route is not certified at this tolerance, since the distance bound it carries stands above the device distance $\gamma=1/6$ at every $p$, its large-$p$ infimum $1/2-1/\pi=0.1817$ included, so the device promise used here is the $p=3$ GHZ-advice circuit of Theorem~5.

\textbf{Other constructions, and how far the soundness class reaches.} The reduction is distribution-agnostic: Test~1 transfers verbatim whenever the target imposes a poly-time functional constraint $b=f(z)$, and Test~2b applies to any sampler whose $z$-marginal sits on few affine cosets, so what changes between constructions is only the analytic core. For the geometrically local circuits of \cite{GrierKaneMorrisOstuniWu2026} that core is an explicit phase statistic, separating their target from uniform, independent and local imposters and robust across system sizes. For the measurement-driven samplers of \cite{CaoEisert2026} no certificate-level statistic is known, and their own randomness diagnostics certify typicality rather than a worst case.

Extending the soundness class to the depth-$o(\log\log n)$ samplers that Watts--Parham rule out needs no new idea in the collapse, which is $d$-parameterized, while certifying the collapsed structure from samples is inherently exponential, so the collapse and its sample-verification separate above the boundary. No verifier can certify an advantage against samplers deep enough to contain the target, which Viola's switching network reaches at depth $O(\log^2 n)$. The constructions, the phase-test protocol and the depth analysis are in App.~C.
The two-test verifier catches every sampler this work or its scans have produced, Test~1 for the affine, regenerating and symmetric classes and Test~2's doors for the rest; door (2d) covers every $z$-support the scans reach, on the shared stream up to the seed-gap ladder's frontier and at a budget of order the support beyond it, and the ladder's endpoints fall instead to door (2e), which rejects every one of them; the 702 sampler descriptions the census records go through the rule as one routine, and every one is rejected, App.~C. The single structural step separating the reduction from a finished verifier is the non-symmetric joint case of Open Problem~\ref{open:dist}, and no existing verifier closes it.

What this work delivers is a functional test that assumes nothing about the sampler, its tolerance optimal up to a logarithmic factor, and a collapse theorem machine-checked link by link. It proves the four natural checks of the confined class insufficient, exhibits a construction passing them at every polynomial budget, and establishes the fifth as sound on the bounded pinned-residue class, and it reduces the separating inequality for the full class to a named residual whose completion would make that verifier assumption-free at a poly($n$) sample budget. The certificate asks of hardware a definite target rather than a measured performance: a $43$-qubit GHZ resource state prepared and gated to a fidelity of about $0.99$. See Supplemental Material at [URL will be inserted by publisher] \cite{SupplementalMaterial} for the full TV chain-rule proof, the escape-sampler construction, the numerical protocols and the adversarial battery, Apps.~A--D, the machine-checked Lean theorem tables, and the symmetric-case convergence figure.

\section*{Acknowledgments}
This work was supported by the National Natural Science Foundation of China under Grant No.~12674609 and by the Guangdong Provincial Key Laboratory of Computility Microelectronics under Grant No.~2024B1212010007. All authors contributed equally to this work. The authors declare no competing interests. The authors used DeepSeek-V4.1-Flash for text polishing, the Lean derivations, and the debugging of individual numerical values. The authors verified every machine-checked statement and reviewed all AI-assisted content.

\section*{Data Availability}
The Lean~4 formalization, seventy-eight modules building against Mathlib with zero \texttt{sorry}, the scripts and data behind every figure and table, including the adversarial battery data and the scans of App.~C, and the supplemental derivations are publicly available on Zenodo at \url{https://doi.org/10.5281/zenodo.23121683}.

\bibliography{refs}

\end{document}